\documentclass[a4paper,UKenglish,cleveref, autoref, thm-restate]{lipics-v2021}
\pdfoutput=1 %uncomment to ensure pdflatex processing (mandatatory e.g. to submit to arXiv)
\hideLIPIcs  %uncomment to remove references to LIPIcs series (logo, DOI, ...), e.g. when preparing a pre-final version to be uploaded to arXiv or another public repository

\usepackage{xspace}
\usepackage{booktabs}
\usepackage[ruled,vlined,linesnumbered]{algorithm2e}

\usepackage{tikz}
\usetikzlibrary{positioning, calc}

\usepackage{subcaption}
\usepackage{complexity}

\newcommand{\sentence}{\Gamma}
\newcommand{\formula}{\varphi}

\newcommand{\structure}{\mathcal{A}}

\newcommand{\fotwo}{\ensuremath{\mathbf{FO}^2}\xspace}

\newcommand{\ctwo}{\ensuremath{\mathbf{C}^2}\xspace}

\newcommand{\domain}{\Delta}

\newcommand{\vecn}{\mathbf{n}}

\newcommand{\utypes}{\ensuremath{\mathcal{U}}\xspace}
\newcommand{\btypes}{\ensuremath{\mathcal{B}}\xspace}

\newtheorem*{problem}{Problem}

\title{On Knowledge Compilation For Two-Variable First-Order Logic} %TODO Please add

\author{Qiaolan Meng}{Beihang University, China}{mengql1606@buaa.edu.cn}{https://orcid.org/0009-0003-3413-5421}{}

\author{Juhua Pu}{Beihang University, China}{pujh@buaa.edu.cn}{https://orcid.org/0000-0003-3866-8703}{}

\author{Hongting Niu}{Beihang University, China}{niuhongting@buaa.edu.cn}{https://orcid.org/0000-0001-5005-2317}{}

\author{Yuyi Wang}{CRRC Zhuzhou Institute, China}{yuyiwang920@gmail.com}{https://orcid.org/0000-0001-7273-9873}{}

\author{Yuanhong Wang}{Jilin University, China}{lucienwang@jlu.edu.cn}{https://orcid.org/0000-0003-1801-9634}{}

\author{Ondřej Kuželka}{Czech Technical University in Prague, Czech Republic}{ondrej.kuzelka@fel.cvut.cz}{https://orcid.org/0000-0002-6523-9114}{}

\authorrunning{Q. Meng et al.} %TODO mandatory. First: Use abbreviated first/middle names. Second (only in severe cases): Use first author plus 'et al.'

\Copyright{Qiaolan Meng, Juhua Pu, Hongting Niu, Yuyi Wang, Yuanhong Wang, and Ondřej Kuželka} %TODO mandatory, please use full first names. LIPIcs license is "CC-BY";  http://creativecommons.org/licenses/by/3.0/

 %TODO mandatory; please add comma-separated list of keywords

\category{} %optional, e.g. invited paper

\acknowledgements{We thank the anonymous reviewers for their helpful comments and suggestions.}%optional

\nolinenumbers %uncomment to disable line numbering

\EventEditors{Alexey Ignatiev and Stefan Szeider}
\EventNoEds{2}
\EventLongTitle{29th International Conference on Theory and Applications of Satisfiability Testing (SAT 2026)}
\EventShortTitle{SAT 2026}
\EventAcronym{SAT}
\EventYear{2026}
\EventDate{July 20--23, 2026}
\EventLocation{Lisbon, Portugal}
\EventLogo{}
\SeriesVolume{377}
\ArticleNo{32}
\begin{document}

\maketitle

\begin{abstract}
Knowledge compilation transforms logical theories into circuit representations that support efficient reasoning. 
We study this problem for propositional groundings of \fotwo{}, the two-variable fragment of first-order logic over finite domains. 
Given an \fotwo{} sentence and a domain of size $n$, its grounding yields a propositional theory over ground atoms. 
We ask whether such theories admit compact representations in DNNF-based and related knowledge compilation languages, and whether these can be constructed efficiently, both with respect to the domain size $n$ for a fixed sentence.
We show first that compact compilation is impossible in general: there exists an \fotwo{} sentence whose grounding over a domain of size $n$ requires DNNF size $2^{\Omega(n)}$. 
On the positive side, we develop a two-stage compiler that exploits the symmetries inherent in the propositional groundings of \fotwo{} sentences.
It branches on unary and binary types rather than individual ground atoms, in a similar spirit to lifted inferences for probabilistic relational models.
Moreover, it optimizes the compilation process by efficiently identifying and caching residual subproblems that are equivalent with respect to future extensions.
Experiments show the practical efficiency of our approach, which often produces smaller circuits and compiles faster than straightforward grounding-based baselines.
\end{abstract}

\section{Introduction}
\label{introduction}

Knowledge compilation (KC) studies how to transform a logical theory into a target representation that supports efficient downstream reasoning. 
Prominent target languages include the decomposable negation normal form (DNNF), which admits tractable queries such as satisfiability and conditioning, and its deterministic variant d-DNNF, which additionally supports efficient model counting and related inference tasks~\cite{darwiche2002knowledge,darwiche2001decomposable}. 
% A central question in this area is whether a given family of theories admits compact representations in such languages.

In this paper, we study knowledge compilation for groundings of first-order theories, focusing on the two-variable fragment \fotwo over finite domains. 
Given a fixed \fotwo sentence $\sentence$ and a finite domain $\domain$, grounding $\sentence$ over $\domain$ yields a propositional theory whose variables correspond to ground atoms. We ask two basic questions: whether these grounded theories admit compact representations in DNNF-based languages, and whether such representations can be constructed efficiently, both with respect to the domain size $n$ for a fixed sentence.
% This is a natural compilation problem: the source description is first-order and highly symmetric, while the target is a standard propositional language for exact reasoning.

Answering these questions has immediate practical value.
For instance, d-DNNF admits linear-time conditioning~\cite{darwiche2002knowledge} and model sampling~\cite{gupta2019weighted, lai2022ccdd}, which naturally enables exact conditional sampling to complete partially observed relational data while strictly respecting universal and existential constraints. 
Furthermore, when the compiled circuit is organized as a structured d-DNNF with respect to a fixed variable tree, it can be composed with other structured circuits over the same variables, enabling modular integration with probabilistic or neurosymbolic architectures~\cite{ahmed2022semantic,AhmedAISTATS23,AhmedNeurIPS23}.
Particularly, neurosymbolic learning systems usually reason over complex relational domains, and their symbolic knowledge bases are inherently specified in first-order logic (FOL) rather than propositional logic. 
Targeting the knowledge compilation of the $\fotwo$ fragment, a well-studied subset of FOL with balanced expressivity and decidability~\cite{1997decision}, offers a vital path to hybrid logic and learning pipelines at scale.

% Answering these questions is worthwhile even though \fotwo{} admits efficient lifted algorithms for tasks such as symmetric weighted model counting~\cite{beame2015symmetric} and sampling~\cite{wang_exact_2023-1}.
% Lifted algorithms exploit the exchangeability present in the first-order specification and are especially effective when weights and evidence preserve this symmetry. 
% Compilation into propositional d-DNNF serves a complementary purpose. 
% Once the grounding of an \fotwo{} sentence has been compiled, the same circuit can be reused for weighted model counting with arbitrary weights on individual ground literals, rather than only predicate-level weights handled by symmetric weighted first-order model counting. 
% % This is important in applications where weights are obtained from heterogeneous data sources or learned predictors and may vary from one tuple to another.
% Furthermore, d-DNNF preserves efficient counting and sampling after conditioning~\cite{gupta2019weighted, lai2022ccdd}, and can handle arbitrary evidence on ground atoms, including evidence on binary literals, even when such evidence destroys the symmetries exploited by lifted inference. 
% Finally, when the compiled circuit is organized as a structured d-DNNF with respect to a fixed variable tree, it can be composed with other structured circuits over the same variables, enabling modular integration with probabilistic or neurosymbolic architectures~\cite{ahmed2022semantic,AhmedAISTATS23,AhmedNeurIPS23}.

Knowledge compilation for first-order theories differs from both lifted inference~~\cite{van_den_broeck_lifted_2011,beame2015symmetric,crane2} and propositional compilation~\cite{c2d,d4}.
Typically, lifted inference targets first-order theories and performs inference without grounding.
For instance, lifted inference for the full \fotwo{} fragment is tractable in the \emph{symmetric} setting,
% It can only apply to the \emph{symmetric} setting, 
where the weight of models is invariant under permutations of the domain elements, and it does \emph{not} support conditioning on binary evidence that breaks the symmetries~\cite{van_den_broeck_conditioning_2012-1}.
In contrast, propositional compilation targets propositional theories, where conditioning on arbitrary evidence is supported, but the input theory is already grounded, so the symmetry and structural regularity present in the first-order description are no longer explicit.
Though there are a few works on exploiting the implicit symmetries in propositional model counting~\cite{wang2022symmc,van2021symmetric}, how to leverage the symmetries in knowledge compilation remains largely unexplored.
Therefore, this paper lies between these two viewpoints: we start from a first-order specification, but aim to compile its propositional grounding into a standard propositional target language.

Our contributions are both negative and positive. 
On the negative side, we show that polynomial-size compilation is impossible in general: there exists an \fotwo sentence whose grounding over a domain of size $n$ requires DNNF size $2^{\Omega(n)}$. 
On the positive side, we develop a specialized two-stage compiler for \fotwo sentences. 
The compiler exploits the fact that \fotwo models over finite domains can be organized via finitely many unary and binary types, and therefore branches on types rather than on individual ground atoms. 
We further introduce an optimization based on equivalent residual subproblems, allowing the compiler to reuse compiled subcircuits when different contexts admit the same set of future completions. 
% Our experiments show that this approach is effective on several benchmark families and often produces substantially smaller circuits than straightforward grounding-based baselines.
Our experiments on several benchmark families suggest that this approach can produce smaller circuits than straightforward grounding-based baselines.

The paper is structured as follows.
Section~\ref{sec:prelim} recalls the necessary background. 
Section~\ref{sec:hardness} proves the exponential lower bound for \,\fotwo{} compilation. 
% Section~\ref{sec:types} introduces types and configurations, and Sections~\ref{sec:overview}--\ref{sec:compatibility} present the two-stage compilation and efficient extendability checks. 
Section~\ref{sec:compilation} presents the two-stage compilation procedure for \fotwo{}.
Section~\ref{sec:optimizations} develops the optimization techniques, and
Section~\ref{sec:experiments} reports an empirical evaluation.

\section{Related Work}
\label{related}
% Knowledge compilation trades succinctness for tractable queries, and the compilation map of~\cite{darwiche2002knowledge} is the standard reference for positioning target languages.
% DNNF and d-DNNF are central due to their balance between expressiveness and polynomial-time support for many queries~\cite{darwiche2001decomposable,darwiche2002knowledge}. Propositional compilation into (d-)DNNF has been extensively studied~\cite{darwiche1999compiling,c2d,pipatsrisawat2010top,d4,bella}, alongside strong lower bounds and separations~\cite{Bova2014ASE,de2020lower,pipatsrisawat2010lower}, often via communication-complexity techniques~\cite{bova2016knowledge}. 
% Moving beyond propositional theories, first-order knowledge compilation introduces lifted circuit representations (e.g., FO-d-DNNF) that support a more limited query set~\cite{van_den_broeck_lifted_2011,van_den_broeck_conditioning_2012-1}.
% In contrast, our work also targets first-order theories but outputs a fully propositional d-DNNF, thereby supporting the standard d-DNNF queries.

Knowledge compilation studies how to trade succinctness for tractable queries by compiling logical theories into target languages with polynomial-time support for selected operations. 
The compilation map in \cite{darwiche2002knowledge} remains the standard reference for positioning such languages. 
Among them, DNNF and d-DNNF are central because they strike a useful balance between succinctness and tractable querying~\cite{darwiche2001decomposable}.
% including satisfiability, model counting, conditioning, and related operations
Other languages include OBDDs, FBDDs, and SDDs~\cite{darwiche2011sdd}, which are more restrictive but support additional queries such as equivalence checking and validity checking in polynomial time~\cite{darwiche2002knowledge}.
On the propositional side, compilation into variants of DNNF has been studied extensively, both from the algorithmic perspective~\cite{darwiche1999compiling,c2d,pipatsrisawat2010top,d4,bella} and from the perspective of lower bounds and separations~\cite{Bova2014ASE,de2020lower,pipatsrisawat2010lower}, often using communication-complexity techniques~\cite{bova2016knowledge}. 

Moving beyond propositional theories, first-order knowledge compilation and lifted probabilistic inference introduce lifted circuit representations, such as FO-d-DNNF, that operate directly on first-order structure and support lifted inference~\cite{van_den_broeck_conditioning_2012-1, van_den_broeck_lifted_2011, crane2}. 
These approaches differ from ours both in the target representation and in the intended reasoning regime. 
Our goal is to compile an \fotwo{} sentence and a finite domain into a standard propositional d-DNNF circuit, thereby retaining the query support of propositional d-DNNF.

Our work is also related to recent results on algorithmic tasks over \fotwo{} models, such as symmetric weighted first-order model counting~\cite{beame2015symmetric}, lifted first-order model sampling~\cite{wang2024lifted}, and model enumeration for two-variable logic~\cite{meng2025model}. 
These works show that \fotwo{} admits unusually strong algorithmic structure for counting, sampling, and enumeration. 
Our focus is different: we study compilation into DNNF-based circuit languages, including both lower bounds on representation size and a specialized compilation procedure. 
At the same time, these adjacent lines of work help motivate our setting, since compiled d-DNNF circuits can serve as a basis for exact reasoning tasks such as counting and sampling on grounded \fotwo{} theories.
% counting, conditioning, enumeration, and conditional sampling

\section{Preliminaries}
\label{sec:prelim}

We recall basic notions from first-order logic and knowledge compilation, and fix notations used throughout the paper.

\subsection{Propositional and First-Order Logic}
\label{sec:prelim:fol}

We assume familiarity with standard notions from propositional and first-order logic~\cite{libkin2004elements}.
In this paper, we work on the \emph{function-free} \emph{finite-domain} fragment of first-order logic.
A vocabulary is a finite set $\mathcal{P}$ of predicate symbols. 
An \emph{atom} is of the form $P(t_1,\dots,t_k)$ where $P/k\in\mathcal{P}$ and each $t_i$ is a variable or a constant.
A \emph{literal} is an atom or its negation. 
Formulas are built from atoms using $\neg,\wedge,\vee$ and may be prefixed by quantifiers $\forall x$ and $\exists x$. 
A formula is a \emph{sentence} if all variables are bound. 
We write $\mathcal{P}_\varphi$ for the vocabulary occurring in a formula $\varphi$. 
A formula is \emph{ground} if it contains no variables, i.e., it is propositional where ground atoms are viewed as boolean variables.
Given a first-order sentence $\sentence$ and a domain $\domain$, \emph{grounding} of $(\sentence, \domain)$ is the propositional formula obtained by instantiating all quantifiers in $\sentence$ over elements from $\domain$ with the usual semantics.
Unless otherwise specified, the first-order formulas in this paper are assumed to be \emph{constant-free}, i.e., they contain no constant symbols.

\begin{example}
Consider the sentence $\sentence = \forall x : \neg E(x,x) \wedge \forall x \exists y : E(x,y)$ and the vertex domain $\domain=\{e_1,e_2\}$.
The grounding of $(\sentence,\domain)$ is the propositional formula
\[(E(e_1,e_1)\vee E(e_1,e_2))
\wedge
(E(e_2,e_1)\vee E(e_2,e_2))
\wedge
\neg E(e_1,e_1)
\wedge
\neg E(e_2,e_2),
\]
where each ground atom, such as $E(e_1,e_2)$, is a propositional variable.
\end{example}

For a vocabulary $\mathcal{P}$ and a finite domain $\domain$,
A \emph{structure} $\structure$ is a tuple $(\domain, f)$, where $f$ interprets each predicate $P/k\in\mathcal P$ as a relation $f(P)\subseteq \domain^k$.
In this paper, a structure is often written as a set of ground literals interpreted by $f$, that is, $\structure = \bigcup_{P/k \in \mathcal{P}} \{P(a_1,\dots,a_k) \mid (a_1,\dots,a_k) \in f(P)\} \cup \{\neg P(a_1,\dots,a_k) \mid (a_1,\dots,a_k) \notin f(P)\}$.
The satisfaction relation $\models$ is defined as usual.
A \emph{model} $\structure$ of a sentence $\sentence$ over a domain $\domain$ is a structure with domain $\domain$ such that $\structure \models \sentence$.

\subsection{Knowledge Compilation}
\label{sec:prelim:kc}

Knowledge compilation (KC) is the process of transforming a propositional formula into a representation such that various logical queries on the formula can be answered efficiently by operating on the compiled representation.
In this paper, we focus on compilation into negation normal form and its variants~\cite{darwiche2002knowledge}.

A negation normal form (NNF) is a node labeled directed acyclic graph (DAG) with internal nodes labeled by logical connectives (e.g., $\wedge,\vee$), leaves labeled by boolean variables or their negations (literals), and a unique sink node (outdegree $0$) called the output.
For an NNF $C$, $\mathrm{var}(C)$ denotes its set of variables and $|C|$  denotes its size, i.e., the number of edges.

Given an NNF $C$ and an assignment $f$ of truth values to its variables, the evaluation of $C$ under $f$, denoted $C(f)$, is defined as the value computed at the output node when the leaves are assigned according to $f$ and internal nodes are evaluated according to their labels in a bottom-up fashion.
The assignments $f$ such that $C(f)=1$ are called the \emph{satisfying assignments} of $C$.
Two NNFs are \emph{logically equivalent} (or simply \emph{equivalent}) if they have the same set of satisfying assignments.
An NNF $C$ \emph{encodes} a propositional formula $\formula$ if $C$ and $\formula$ have the same set of satisfying assignments.

An NNF circuit $C$ is \emph{decomposable} (DNNF) if for every $\wedge$-node with children $C_1,\dots,C_r$, the boolean variable sets are pairwise disjoint:
$\mathrm{var}(C_i)\cap \mathrm{var}(C_j)=\emptyset$ for all $i\neq j$.
$C$ is \emph{deterministic} if for every $\vee$-node with children $C_1,\dots,C_r$, the disjuncts are pairwise inconsistent, i.e., $C_i\wedge C_j$ is unsatisfiable for all $i\neq j$.
A circuit satisfying both properties is a \emph{deterministic DNNF} (d-DNNF).
d-DNNFs are particularly of interest in knowledge compilation since they support a wide range of queries that can be answered in time polynomial in the size of the circuit~\cite{darwiche2001decomposable,darwiche2002knowledge}.
We also consider the more restrictive class of \emph{structured} d-DNNFs and OBDDs, which are defined with respect to a fixed variable tree and variable order, respectively.
We refer the reader to~\cite{darwiche2002knowledge} for more details on these languages and their properties.

We consider knowledge compilation for first-order logic, focusing primarily on d-DNNFs. 
The same questions can also be asked for more restrictive target languages, and we provide corresponding results for structured d-DNNFs and OBDDs.
% We consider knowledge compilation for first-order logic, where we focus on the d-DNNFs, but the same questions can be asked for structured d-DNNFs, OBDDs, and other target languages.
\begin{problem}
  \label{prob:kc-fo}
  Fix a first-order sentence $\sentence$. Given a finite domain $\domain$, output a d-DNNF circuit encoding the grounding of $(\sentence,\domain)$.
%   \noindent\textbf{Input:} A first-order sentence $\sentence$ and a domain $\domain$.
%   \noindent\textbf{Output:} A d-DNNF circuit encoding the grounding of $(\sentence,\domain)$.
\end{problem}
\noindent Given an algorithm that solves this problem, we mainly care about the size of the output circuit and the time complexity of the algorithm, both w.r.t.\ the domain size $|\domain|$.

% We consider knowledge compilation for first-order logic: input a first-order sentence $\sentence$ and a domain $\domain$, and output a d-DNNF circuit encoding the grounding of $(\sentence,\domain)$.
% Given an algorithm that solves this problem, we mainly care about the size of the output circuit and the time complexity of the algorithm, both w.r.t.\ the domain size $|\domain|$ for a fixed sentence $\sentence$.

\subsection{Two-Variable Logic in Scott Normal Form}
\label{subsec:fo2}
The two-variable fragment of first-order logic, denoted by \fotwo{}, is a well-studied decidable fragment~\cite{mortimer1975languages,1997decision}.
\fotwo{} exhibits particularly favorable properties for various computational tasks, e.g., model counting for \fotwo{} is in time polynomial in the domain size for a fixed sentence, while adding another third variable renders the problem intractable~\cite{beame2015symmetric}.
Therefore, we study knowledge compilation for \fotwo{} as a first step towards understanding the compilation of more expressive first-order fragments.

% In this paper we restrict to \fotwo.
Any \fotwo sentence can be transformed into Scott normal form (SNF)~\cite{Scott1962}:
\begin{align*}
\sentence \;=\; \forall x \forall y: \, \phi(x,y)
\;\wedge\;
\bigwedge_{k\in[m]} \forall x \exists y: \, \psi_k(x,y),
\end{align*}
where $\phi(x,y)$ and each $\psi_k(x,y)$ are quantifier-free and use only variables $x$ and $y$.
For convenience, we further assume that each existential conjunct has been atomized to $\forall x\exists y: \beta_k(x,y)$ where $\beta_k(x,y)$ is a single binary atom.
This is easily achievable by introducing a fresh binary predicate symbol $\beta_k$ for each $\psi_k(x,y)$ and adding universal constraints ensuring $\beta_k(x,y)\leftrightarrow \psi_k(x,y)$.

Throughout the paper, we assume that all \fotwo sentences to be compiled are in SNF.
This assumption is without loss of generality because: 1) the transformation into SNF increases the size of the sentence only linearly; 2) there is a bijective mapping between the models of the original sentence and those of its SNF sentence~\cite{beame2015symmetric,wang2024lifted}, so that the queries (e.g., model counting and enumeration) can be carried out on the SNF sentence and the results can be easily translated back to the original sentence.

\section{Hardness of \texorpdfstring{$\fotwo$}{FO2} Compilation}
\label{sec:hardness}

In this section, we show that there exists an \fotwo sentence whose DNNF representation over a domain of size $n$ cannot be in size polynomial (or even subexponential) in $n$.
This immediately implies the same lower bound for d-DNNF representations.

\begin{theorem}
\label{th:hardness_fo2}
% There exists an $\fotwo$ sentence $\sentence$ such that every DNNF circuit encoding $\sentence$ over a domain of size $n$ has size $2^{\Omega(n)}$.

There exists an $\fotwo$ sentence $\sentence$ such that for every $n\ge 1$ and domain $\domain$ with $|\domain|=n$, every DNNF circuit encoding the grounding of $(\sentence,\domain)$ has size $2^{\Omega(n)}$.
\end{theorem}

The proof relies on the transformation of DNNF circuits under minimization introduced by Darwiche~\cite{darwiche2001decomposable}, and a lower bound on DNNF size for the Boolean function of permutation $\mathsf{PERM}_n$ (or equivalently, the perfect matchings of the complete bipartite graph $K_{n,n}$) established by de Colnet~\cite{de2020lower}.
We recall these results below.

Given an NNF $C$, the \emph{minimization} of $C$ is an NNF whose satisfying assignments are exactly those satisfying assignments of $C$ with the minimum number of true variables. 
\begin{theorem}[Combined from 
\cite{Darwiche01012001} and {\protect\cite[Section~2.5]{darwiche2001decomposable}}]
\label{th:min-poly}
Given a DNNF circuit $C$, one can construct a DNNF circuit $C_{min}$ encoding the minimization of $C$, in time polynomial in $|C|$.
\end{theorem}

The $\mathsf{PERM}_n$ can be represented as a propositional formula over $n^2$ boolean variables $\{p_{ij}\mid i,j\in[n]\}$, and its satisfying assignments of $\mathsf{PERM}_n$ are exactly the $n \times n$ permutation matrices, i.e., there is exactly one true variable in $\{p_{i1},\dots,p_{in}\}$ for each $i\in[n]$ and exactly one true variable in $\{p_{1j},\dots,p_{nj}\}$ for each $j\in[n]$.
\begin{theorem}[\cite{de2020lower}\protect\footnotemark]
\label{th:perm-lb}
Every DNNF circuit encoding $\mathsf{PERM}_n$ has size $2^{\Omega(n)}$.

\footnotetext{Though not stated explicitly, one can easily verify from the proof of Theorem~1 therein.}
\end{theorem}

\begin{proof}[Proof of Theorem~\ref{th:hardness_fo2}]
Consider the \fotwo sentence:
\begin{align*}
\sentence_P
\,=\,
\big( \forall x \exists y: P(x,y) \big)
\, \wedge\,
\big( \forall x \exists y: P(y,x)  \big).
\end{align*}
We prove by contradiction that any DNNF circuit encoding $\sentence_P$ over a domain of size $n$ requires size $2^{\Omega(n)}$.

For a domain $\domain=\{e_1,\dots,e_n\}$, assume for contradiction that there exists a DNNF circuit $C$ encoding the grounding of $(\sentence_P, \domain)$ with size $2^{o(n)}$.
Let $p_{ij}$ be the boolean variable in $C$ corresponding to the ground atom $P(e_i,e_j)$.
The sentence $\sentence_P$ (and thus $C$) enforces that for every $i\in[n]$, there exists $j\in[n]$ such that $p_{ij}$ is true, and for every $j\in[n]$, there exists $i\in[n]$ such that $p_{ij}$ is true.
Let $C_{\min}$ be the minimization of $C$ obtained by Theorem~\ref{th:min-poly}.
One can easily verify that the satisfying assignments of $C_{\min}$ correspond exactly to the permutations of $[n]$, and thus $C_{\min}$ encodes $\mathsf{PERM}_n$.
% Since $C$ encodes the grounding of $\sentence_P$, every satisfying assignment of $C$ has, for each $i\in[n]$, some $j\in[n]$ with $p_{ij}$ true, and for each $j\in[n]$, some $i\in[n]$ with $p_{ij}$ true.
% Therefore every satisfying assignment of $C$ contains at least $n$ true variables.
% Moreover, a satisfying assignment has exactly $n$ true variables iff it contains exactly one true variable in each row and exactly one true variable in each column, i.e., iff it is a permutation matrix. 
% Hence the satisfying assignments of $C_{\min}$ are exactly those of $\mathsf{PERM}_n$, so $C_{\min}$ encodes $\mathsf{PERM}_n$.
% 
From Theorem~\ref{th:min-poly}, $C_{\min}$ is a DNNF circuit of size polynomial in $|C|$, and therefore has size $2^{o(n)}$.
However, by Theorem~\ref{th:perm-lb}, any DNNF circuit encoding $\mathsf{PERM}_n$ requires size $2^{\Omega(n)}$, leading to a contradiction.
\end{proof}

\section{d-DNNF compilation for \texorpdfstring{$\fotwo$}{FO2}}
\label{sec:compilation}

In this section, we present a compilation scheme for \fotwo.
Throughout, we assume that the sentence is in Scott normal form
(cf.\ Section~\ref{subsec:fo2}).
Given a domain $\domain=\{e_1,\dots,e_n\}$, our goal is to compile the grounding of $(\sentence, \domain)$ into a propositional d-DNNF circuit.
We also show how the construction yields OBDD and structured d-DNNF representations.
% The algorithm can be easily adapted to compile into structured d-DNNF or OBDD.

The section is organized as follows.
Section~\ref{sec:types} introduces unary and binary types, which describe the possible local truth assignments of elements and of element pairs, respectively.
These types provide the branching objects used in the compilation.
Section~\ref{sec:overview} then presents a two-stage procedure, where Stage~I assigns unary types to elements, and Stage~II assigns binary types to element pairs. 
Section~\ref{sec:compatibility} then gives an efficient extendability test to prune infeasible branches, which checks whether the current partial type assignment can be extended to a full model.

\subsection{Types}
\label{sec:types}

Let $\mathcal{P}$ be the set of unary and binary predicates in $\sentence$.
We say that a set of literals is \emph{maximally consistent} if it does not contain any contradictory literals and cannot be extended by adding any other literal.
Then we define unary and binary types as follows.

\begin{definition}[Unary and binary types]
\label{def:types}
A \emph{unary type} $\tau$ is a maximally consistent set of unary literals over $\mathcal{P}$ involving a single variable $x$ (including the diagonal literals $R(x,x)$ for binary predicates $R$).
A \emph{binary type} $\pi$ is a maximally consistent set of binary literals over $\mathcal{P}$ involving distinct variables $x$ and $y$.
A type is either a unary type or a binary type.
% for every $R\in\mathcal{P}_2$, $\tau$ also contains exactly one of $R(x,x)$ and $\neg R(x,x)$.
% for every $R\in\mathcal{P}_2$ and each orientation $(x,y)$ and $(y,x)$, exactly one of $R(x,y),\neg R(x,y)$ belongs to $\pi$ and exactly one of $R(y,x),\neg R(y,x)$ belongs to $\pi$
\end{definition}

We often identify a type with a formula obtained by conjoining all its literals, and write $\tau(x)$ and $\pi(x,y)$. 
For a structure $\structure$ and an element $a\in\domain$, we say that $a$ \emph{realizes} a unary type $\tau$ if $\structure\models \tau(a)$, and for two distinct elements $a,b\in\domain$, we say that the ordered pair $(a,b)$ \emph{realizes} a binary type $\pi$ if $\structure\models \pi(a,b)$.
The realization of types for an element or a pair of elements is unique because types are maximally consistent.
% A type is \emph{valid} if it is consistent with the universal constraint $\forall x\forall y:\phi(x,y)$.
Therefore, a model of $\sentence$ over $\domain$ (equivalently, a satisfying assignment of the propositional grounding of $(\sentence,\domain)$) can be viewed as a type assignment that assigns a unary type to each element $e_i$ and a binary type to each ordered pair $(e_i,e_j)$ for $i < j$.

When dealing with types, we only need to consider those that are \emph{valid}.
A unary type $\tau$ is \emph{valid} w.r.t.\ $\sentence$ if
$
\tau(x) \models \phi(x,x),
$
where $\phi$ is from the universal conjunct $\forall x \forall y:\phi(x,y)$ in the Scott normal form of $\sentence$.
Given unary types $\tau$ and $\tau'$, a binary type $\pi$ is \emph{valid} for $(\tau,\tau')$ if
$
\tau(x) \wedge \tau'(y) \wedge \pi(x,y) \models \phi(x,y) \wedge \phi(y,x).
$
Let $\utypes_\sentence$ be the set of valid unary types and, for each pair $(\tau,\tau')\in \utypes_\sentence\times \utypes_\sentence$, let $\btypes_\sentence(\tau,\tau')$ be the set of binary types valid for $(\tau,\tau')$.
When the sentence $\sentence$ is clear from context, we may omit the subscript and write $\utypes$ and $\btypes(\tau,\tau')$ instead.

% Restricting to valid types ensures that the universal constraint $\forall x\forall y\,\phi(x,y)$ are satisfied; we only need to focus on existential constraints, which we handle via compatibility checks below.

\begin{example}
\label{ex:2-color}
Consider the sentence defining 2-colored graphs:
\begin{align*}
  \sentence_{RB} = \ &\forall x: \big(R(x) \vee B(x)\big) \wedge \big(\neg R(x) \vee \neg B(x)\big) \ \wedge  \ \\
  &\forall x\forall y: \big(E(x,y) \rightarrow E(y,x)\big) \wedge \big(E(x,y) \rightarrow ((R(x)\wedge B(y)) \vee (B(x)\wedge R(y)))\big).
\end{align*}
There are $2^3$ unary types in total, corresponding to different combinations of truth values of atoms $R(x),B(x), $ and $E(x,x)$, but only two of them are valid: 
$\tau_R=\{R(x), \neg B(x), \neg E(x,x)\}$ and $\tau_B=\{\neg R(x), B(x), \neg E(x,x)\}$.
For the unary type pairs $(\tau_R,\tau_B)$ and $(\tau_B,\tau_R)$, there are two valid binary types: $\pi_{E}=\{E(x,y), E(y,x)\}$ and $\pi_{\neg E}=\{\neg E(x,y),\neg E(y,x)\}$, while for the pairs $(\tau_R,\tau_R)$ and $(\tau_B,\tau_B)$, only $\pi_{\neg E}$ is valid.

\end{example}

\subsection{Two-Stage Compilation}
\label{sec:overview}

% Our compilation constructs a d-DNNF circuit by branching on \emph{types} rather than on individual ground atoms. 
The construction of the d-DNNF circuit proceeds in two stages.
\begin{itemize}
\item Stage I: unary type assignment.
We process domain elements $e_1,\dots,e_n$ in order. 
At an element $e_i$ we branch over $\tau\in \utypes$, and proceed to $e_{i+1}$. 
When all elements have been processed, the path from the root to the current node corresponds to a complete unary type assignment, and we then
move to Stage II.

\item Stage II: binary type assignment.
Based on a unary assignment, we process ordered pairs $(e_i,e_j)$ with $1\le i<j\le n$ in lexicographic order of $(i,j)$.
At each pair $(e_i,e_j)$ we branch over $\pi\in \btypes(\tau_i,\tau_j)$, 
and continue to the next pair until all pairs have been processed.
\end{itemize}

We implement this two-stage compilation as a recursive procedure, as shown in Algorithm~\ref{alg:two-stage-algorithm}.
Each call \textsc{Stage-I}$(i,\theta)$ returns the output node of a subcircuit whose satisfying assignments are exactly the completions of the current context $\theta$, where the context $\theta$ is formally defined as follows.
\begin{definition}[Context]\label{def:context}
The context $\theta$ is a propositional formula 
% over the grounding variables that is 
obtained by conjoining the branch constraints accumulated along the current recursive call.
\end{definition}

Intuitively, the context $\theta$ captures the constraints imposed by the type decisions made so far and thus represents a partial type assignment.
Although branching on valid types ensures the universal conjunct
$\forall x\forall y:\phi(x,y)$, some branches may still violate the existential constraints $\forall x\exists y:\beta_k(x,y)$ in $\sentence$.
We prune such branches by checking whether their contexts can be extended to full models of $\sentence$ over $\domain$.
We call this the \textsc{Extendable} test and discuss its details in Section~\ref{sec:compatibility}.

\begin{example}
\label{ex:context}
Consider the sentence defining graphs without isolated vertices:
\begin{align*}
\sentence_{E} = \big(\forall x: \neg E(x,x)\big) \wedge  \big(\forall x\forall y: E(x,y) \rightarrow E(y,x)\big) \wedge \big(\forall x\exists y: E(x,y)\big).
\end{align*}
\noindent
We compile $\sentence_{E}$ over the vertex domain $\domain=\{e_1,e_2,e_3\}$.
There is only one valid unary type $\tau=\{\neg E(x,x)\}$.
Hence Stage~I produces the unique context
$
\theta=\tau(e_1)\wedge \tau(e_2)\wedge \tau(e_3).
$
In Stage~II, suppose we assign to the pair $(e_1,e_2)$ the binary type
$
\pi_{\neg E}=\{\neg E(x,y),\neg E(y,x)\},
$
which updates the context to
$
\theta'=\theta\wedge \pi_{\neg E}(e_1,e_2).
$
When the next pair $(e_1,e_3)$ is processed, the branch corresponding to $\pi_{\neg E}$ can be pruned, because
$
\theta' \wedge \pi_{\neg E}(e_1,e_3)
$
is not extendable to any full model of $\sentence_{E}$ over $\domain$.
Indeed, $e_1$ is adjacent to neither $e_2$ nor $e_3$ under this context, and the existential requirement $\exists y\,E(e_1,y)$ cannot be satisfied.
\end{example}

The function \textsc{Stage-I}$(i,\theta)$ creates a set $\mathcal{S}$ to collect the subcircuits for each extendable unary type branch at element $e_i$, and returns an $\vee$-node over them.
For each unary type $\tau$, it recursively calls \textsc{Stage-I}$(i{+}1,\theta \wedge \tau(e_i))$ to compile the subcircuit for the remaining elements, and conjoins it with an $\wedge$-node over the literal leaves in $\tau(e_i)$.
Then all such $\wedge$-nodes are combined into an $\vee$-node and returned.
\cref{fig:or-node-types} illustrates the decision structure used in Stage I when branching on unary types of an element in~\cref{ex:2-color}.
When all elements have been processed, it invokes \textsc{Stage-II}$((1,2),\theta)$, which works analogously.
The function $\text{next}(i,j)$ returns the next pair to be processed after $(e_i,e_j)$ in lexicographic order.
When all pairs have been processed, the final output of \textsc{Stage-I}$(1,\top)$ is the compiled circuit $C_{\sentence,\domain}$.

\begin{figure}[t]
\centering
\begin{tikzpicture}
[
    op/.style={draw, circle, inner sep=1.3pt, minimum size=4.5mm},
    leaf/.style={draw, inner sep=2.0pt, rounded corners, font=\footnotesize},
]

  % --- Level 0: Root ---
  \node[op] (root) {$\vee$};

  % --- Level 1: First AND layer ---
  \node[op, below left=0.2cm and 1.0cm of root] (typeRed) {$\wedge$};
  \node[op, below right=0.2cm and 1.0cm of root] (typeBlue) {$\wedge$};

  \draw (root) -- (typeRed);
  \draw (root) -- (typeBlue);

  % --- Level 2: Branch Components ---

  \node[leaf, below=1.0cm of root] (nE) {$\neg E(e_i,e_i)$};

  \node[leaf, left=0.3cm of nE] (nBlue) {$\neg B(e_i)$};
  \node[leaf, left=0.3cm of nBlue] (Red) {$R(e_i)$};
  \node[op, left=0.6cm of Red] (veL) {$\vee$};

  \node[leaf, right=0.3cm of nE] (nRed) {$\neg R(e_i)$};
  \node[leaf, right=0.3cm of nRed] (Blue) {$B(e_i)$};
  \node[op, right=0.6cm of Blue] (veR) {$\vee$};

  \draw (typeRed) -- (nE);
  \draw (typeRed) -- (veL);
  \draw (typeRed) -- (Red);
  \draw (typeRed) -- (nBlue);

  \draw (typeBlue) -- (nE);
  \draw (typeBlue) -- (veR);
  \draw (typeBlue) -- (nRed);
  \draw (typeBlue) -- (Blue);
\end{tikzpicture}
\caption{When branching on $e_i$ at Stage~I in~\cref{ex:2-color}, we have two disjuncts corresponding to the two valid unary types $\tau_R$ and $\tau_B$. The two $\vee$-nodes at the bottom connect to the subcircuits for the remaining elements.}
\label{fig:or-node-types}
\end{figure}
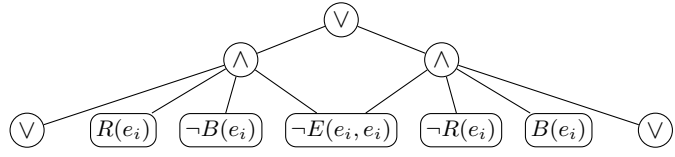

\begin{algorithm}[tb]
    \caption{Two-Stage Compilation}
    \label{alg:two-stage-algorithm}
    
    \SetKwInOut{Input}{Input}
    \SetKwInOut{Output}{Output}

    \SetKwProg{Fn}{Function}{ is}{end}
    
    \Input{A fixed sentence $\sentence$ and a finite domain $\{e_1,\dots,e_n\}$}
    \Output{A d-DNNF circuit}
    
    \tcc{$\textsc{Or}(\cdot)$ and $\textsc{And}(\cdot)$ create $\vee$-nodes and $\wedge$-nodes respectively.}
    
    \BlankLine
    \Return \textsc{Stage-I}(1, $\top$)\;
    
    \BlankLine
    \Fn{\textnormal{\textsc{Stage-I}$ (i, \theta)$}}{
        \If{$i=n+1$}{
            \Return \textsc{Stage-II}((1, 2), $\theta$)\;
        }
        $\mathcal{S} \gets \emptyset$\;
        \ForEach{$\tau \in \utypes$}{
            \If{\textnormal{\textsc{Extendable}$(\theta \wedge \tau(e_i))$}}{
                $C \gets \textsc{Stage-I}(i+1, \theta \wedge \tau(e_i))$\label{line:stage-i-recursion}\;
                $\mathcal{S} \gets \mathcal{S} \cup \{\textsc{And}(\textsc{And}(\tau(e_i)), C) \}$
            }
        }
        \Return \textsc{Or}($\mathcal{S}$) 
    }
    
    \BlankLine
    \Fn{\textnormal{\textsc{Stage-II}$((i, j), \theta)$}}{
        \If{$i=n$}{
            \Return $\top$\;
        }
        $\mathcal{S} \gets \emptyset$\;
        \ForEach{$\pi \in \btypes(\tau_i, \tau_j)$}{
            \If{\textnormal{\textsc{Extendable}$(\theta \wedge \pi(e_i, e_j))$}}{
                $C \gets \textsc{Stage-II}(\text{next}(i, j), \theta \wedge \pi(e_i, e_j))$\label{line:stage-ii-recursion}\;
                $\mathcal{S} \gets \mathcal{S} \cup \{\textsc{And}(\textsc{And}(\pi(e_i, e_j)), C) \}$
            }
        }
        \Return \textsc{Or}($\mathcal{S}$)
    }
\end{algorithm}

We show that the two-stage compilation constructs a d-DNNF circuit that encodes the propositional grounding of $(\sentence,\domain)$ as stated in the following theorem.

\begin{theorem}
\label{th:two-stage-correct}
Fix an \fotwo sentence $\sentence$ and a domain $\domain=\{e_1,\dots,e_n\}$.
Let $\Phi_{\sentence,\domain}$ be the propositional grounding of $(\sentence,\domain)$.
Algorithm~\ref{alg:two-stage-algorithm} outputs a d-DNNF that encodes $\Phi_{\sentence,\domain}$.
\end{theorem}

\begin{proof}[Proof]

Denote by $C_{\sentence,\domain}$ the output circuit of Algorithm~\ref{alg:two-stage-algorithm}.
We first show that every $\wedge$-node in $C_{\sentence,\domain}$ is decomposable, and every $\vee$-node in $C_{\sentence,\domain}$ is deterministic, thereby establishing that $C_{\sentence,\domain}$ is a d-DNNF circuit.

The $\wedge$-nodes are created only when branching on a unary type $\tau$ for an element $e_i$ or a binary type $\pi$ for a pair $(e_i,e_j)$, where the two conjuncts are the conjunction of literals in $\tau(e_i)$ or $\pi(e_i,e_j)$ and the $\vee$-node for the next element or pair.
When processing $e_i$ and branching on a unary type $\tau$, the $\wedge$-node created for this branch consists of two parts: the conjunction of literals in $\tau(e_i)$ and recursive subcircuit for the next element $e_{i+1}$.
Since the literals in $\tau(e_i)$ involve only unary atoms and diagonal atoms on $e_i$, while the recursive subcircuit never branches on such atoms again, the two conjuncts have disjoint variable sets, making the $\wedge$-node decomposable.
A similar argument applies to $\wedge$-nodes created when branching on a binary type $\pi$ for a pair $(e_i,e_j)$.
Thus every $\wedge$-node is decomposable.

$\vee$-nodes are created only when branching on unary types for an element $e_i$ and binary types for a pair $(e_i,e_j)$.
Consider an $\vee$-node created when branching on unary types for an element $e_i$.
For any two distinct unary types $\tau\neq\tau'$, there exists an atom $a$ such that $\tau(e_i)$ contains $a$ while $\tau'(e_i)$ contains $\neg a$, making $\tau(e_i)\wedge\tau'(e_i)$ unsatisfiable.
Therefore, the disjuncts at this $\vee$-node are mutually exclusive, and the node is deterministic.
The same argument applies to $\vee$-nodes created when branching on binary types for a pair $(e_i,e_j)$.
Thus every $\vee$-node is deterministic.

Then we show that $C_{\sentence,\domain}$ encodes $\Phi_{\sentence,\domain}$, i.e., they have the same satisfying assignments. 
Any satisfying assignment of $C_{\sentence,\domain}$ selects exactly one branch at each deterministic decision node, and thus determines a complete type assignment. 
Since the construction branches only on valid types and prunes every branch whose context cannot be extended to a full model of $\sentence$ over $\domain$, this assignment can be extended to a model of $\sentence$.
As all unary and binary types are fixed at a leaf, the induced grounding assignment itself satisfies $\Phi_{\sentence,\domain}$.
Conversely, any satisfying assignment of $\Phi_{\sentence,\domain}$ induces a complete valid type assignment that is extendable at every stage, so the corresponding branch is preserved by Algorithm~\ref{alg:two-stage-algorithm} and hence satisfies $C_{\sentence,\domain}$. 
Therefore, $C_{\sentence,\domain}$ encodes $\Phi_{\sentence,\domain}$.
\end{proof}

We observe that~\cref{alg:two-stage-algorithm} processes the propositional variables according to a natural block order. 
It first processes the unary types of the elements $e_1,\ldots,e_n$, and then the binary types of pairs $(e_i,e_j)$ with $i<j$ in lexicographic order. 
Let $B_i$ be the block of variables corresponding to the unary atoms and diagonal atoms on $e_i$, and let $B_{ij}$ be the block of variables corresponding to the binary atoms on $(e_i,e_j)$. 
\cref{alg:two-stage-algorithm} induces a variable order that respects the block order $B_1, \dots, B_n, B_{12}, \dots, B_{(n-1)n}$.
Further note that the variables in each block are only involved in the branching at one stage (Stage I for $B_i$ and Stage II for $B_{ij}$) and are not branched on at any other stage.
If we view the branching on a block as a single branching step, then the resulting circuit is essentially an OBDD with respect to the block order.
% Inside each block, we fix an arbitrary order of the variables.
% Then we have the following corollary on the output circuit.

\begin{proposition}
\label{cor:OBDD}
The circuit output by Algorithm~\ref{alg:two-stage-algorithm} can be transformed into an OBDD respecting the variable order induced by the above block order, with only a constant-factor size overhead.
\end{proposition}

\begin{proof}
For convenience, we denote the block order induced by~\cref{alg:two-stage-algorithm} as $b_1,\dots,b_N$, and fix an arbitrary order of the variables inside each block.
We call a recursive call (either \textsc{Stage-I} or \textsc{Stage-II}) being at level $t$ if it processes block $b_t$, and the terminal calls are at level $N+1$.
By a backward induction on $t$, one can easily show that any subcircuit returned by a call at level $t$ mentions only variables in $b_t,\dots,b_N$.
% , since the terminal calls return constants and the non-terminal calls only branch on types over the current block $b_t$ before recursively processing later blocks.

Consider a call at level $t$.  
The returned subcircuit must be in the form $\bigvee_{\alpha\in A_t}\bigl(\chi_\alpha\wedge C_\alpha\bigr)$,
where $A_t$ is the set of retained valid types for the current block $b_t$, $\chi_\alpha$ is the conjunction of literals specifying the type $\alpha$, and each $C_\alpha$ is the subcircuit returned by the recursive call for the next block $b_{t+1}$ under the context extended by $\chi_\alpha$.
Any such subcircuit $\bigvee_{\alpha\in A_t}\bigl(\chi_\alpha\wedge C_\alpha\bigr)$ can be transformed into an OBDD fragment that tests the variables of $b_t$ in the fixed within-block order: assignments corresponding to $\alpha$ are routed to the OBDD for $C_\alpha$, and all other assignments are routed to $0$.  
Applying this transformation to the subcircuit returned by each call at level $t$ for $t=N,\dots,1$ results in an OBDD respecting the variable order induced by the block order.

Finally, we analyze the size overhead of this transformation. 
Since the sentence $\sentence$ is fixed, the variable number in each block is bounded by a constant.  
Therefore each subcircuit $\bigvee_{\alpha\in A_t}\bigl(\chi_\alpha\wedge C_\alpha\bigr)$ can be transformed into an OBDD fragment of size at most a constant factor of the original subcircuit, and thus the overall size overhead is only a constant factor.
\end{proof}

% Since OBDDs are structured d-DNNFs with respect to linear vtrees, 
Since an OBDD respecting a variable order can be represented as a structured d-DNNF respecting the corresponding linear vtree, 
\cref{cor:OBDD} further implies that the output circuit admits a structured d-DNNF with constant-factor overhead for fixed sentence $\sentence$.
\cref{fig:structured-d-DNNF-OBDD} illustrates the structured d-DNNF and OBDD of the example in \cref{fig:or-node-types}.

\begin{figure}[t]
\centering
\begin{subfigure}{0.45\textwidth}
\centering
\begin{tikzpicture}[
    op/.style={draw, circle, inner sep=1.2pt, minimum size=4.3mm},
    leaf/.style={draw, rounded corners, inner sep=2pt, font=\footnotesize}
]

% Root
\node[op] (root) at (0,0) {$\vee$};

% First layer
\node[op] (a1) at (-0.5,-0.7) {$\wedge$};
\node[op] (b1) at ( 0.5,-0.7) {$\wedge$};

\draw (root) -- (a1);
\draw (root) -- (b1);

% Second layer
\node[leaf] (R)  at (-1.35,-1.55) {$R(e_i)$};
\node[op]   (a2) at (-0.5,-1.55) {$\wedge$};

\node[op]   (b2) at ( 0.5,-1.55) {$\wedge$};
\node[leaf] (nR) at ( 1.45,-1.55) {$\neg R(e_i)$};

\draw (a1) -- (R);
\draw (a1) -- (a2);

\draw (b1) -- (b2);
\draw (b1) -- (nR);

% Third layer
\node[leaf] (nB) at (-1.45,-2.45) {$\neg B(e_i)$};
\node[op]   (a3) at (-0.5,-2.45) {$\wedge$};

\node[op]   (b3) at ( 0.5,-2.45) {$\wedge$};
\node[leaf] (B)  at ( 1.35,-2.45) {$B(e_i)$};

\draw (a2) -- (nB);
\draw (a2) -- (a3);

\draw (b2) -- (b3);
\draw (b2) -- (B);

% Fourth layer
\node[op]   (nextL) at (-1.35,-3.3) {$\vee$};
\node[leaf] (nE)    at ( 0.00,-3.3) {$\neg E(e_i,e_i)$};
\node[op]   (nextR) at ( 1.35,-3.3) {$\vee$};

\draw (a3) -- (nextL);
\draw (a3) -- (nE);

\draw (b3) -- (nE);
\draw (b3) -- (nextR);

\end{tikzpicture}
\caption{Structured d-DNNF}
\label{fig:structured-d-DNNF}
\end{subfigure}
\hfill
\begin{subfigure}{0.45\textwidth}
\centering
\begin{tikzpicture}[
    x=1cm, y=1cm, >=latex,
    dnode/.style={
        draw,
        rounded corners=2pt,
        inner xsep=2.5pt,
        inner ysep=1.5pt,
        minimum height=4.5mm,
        font=\footnotesize
    },
    term/.style={
        draw,
        rounded corners=2pt,
        inner sep=2pt,
        minimum width=5mm,
        font=\footnotesize
    },
    sub/.style={
        draw,
        rounded corners=2pt,
        inner sep=2pt,
        minimum width=7mm,
        font=\footnotesize
    },
    elab/.style={
        midway,
        fill=white,
        inner sep=1pt,
        font=\scriptsize
    }
]

% Decision nodes
\node[dnode] (R)  at (0,0) {$R(e_i)$};

\node[dnode] (BL) at (-1.25,-1.05) {$B(e_i)$};
\node[dnode] (BR) at ( 1.25,-1.05) {$B(e_i)$};

\node[dnode] (EL) at (-1.25,-2.25) {$E(e_i,e_i)$};
\node[dnode] (ER) at ( 1.25,-2.25) {$E(e_i,e_i)$};

% Bottom nodes
\node[sub]  (CR) at (-1.25,-3.45) {$C_R$};
\node[term] (Z)  at ( 0.00,-3.45) {$0$};
\node[sub]  (CB) at ( 1.25,-3.45) {$C_B$};

% Edges from R
\draw[->] (R) -- node[elab] {$1$} (BL);
\draw[->] (R) -- node[elab] {$0$} (BR);

% Edges from B nodes
\draw[->] (BL) -- node[elab] {$0$} (EL);
\draw[->] (BR) -- node[elab] {$1$} (ER);

\coordinate (ZL) at ($(Z.north west)!0.7!(Z.north)$);
\coordinate (ZR) at ($(Z.north)!0.3!(Z.north east)$);

\coordinate (BLtoZ) at ($(ZL)+(0,1.2)$);
\coordinate (BRtoZ) at ($(ZR)+(0,1.2)$);

\draw[->] (BL.south east) -- node[elab] {$1$} (BLtoZ) -- (ZL);
\draw[->] (BR.south west) -- node[elab] {$0$} (BRtoZ) -- (ZR);

% Edges from E nodes
\draw[->] (EL) -- node[elab] {$0$} (CR);
\draw[->] (EL) -- node[elab] {$1$} (Z);

\draw[->] (ER) -- node[elab] {$0$} (CB);
\draw[->] (ER) -- node[elab] {$1$} (Z);

\end{tikzpicture}
\caption{OBDD}
\label{fig:OBDD}
\end{subfigure}
\caption{The structured d-DNNF and OBDD of the example in \cref{fig:or-node-types}.}
\label{fig:structured-d-DNNF-OBDD}
\end{figure}
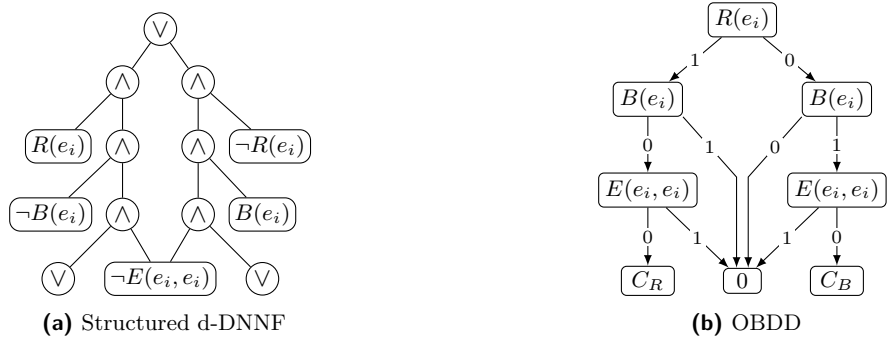

\subsection{Extendability Checks}
\label{sec:compatibility}

At each branching step one must test whether the current context $\theta$, which records the type assignments made so far, can be extended to a full model of $\sentence$ over $\domain$ of size $n$.
We formalize this as follows.
\begin{definition}[Extendability]
  Let $\sentence$ be an \fotwo sentence and let $\theta$ be a context during the compilation over a domain $\domain$. 
  We say that $\theta$ is \emph{extendable} w.r.t. $(\sentence, \domain)$ if there exists a model $\structure$ over $\domain$ such that $\structure \models \sentence \wedge \theta$.
\end{definition}

Intuitively, one can do the extendability check by reducing it to a satisfiability test of the propositional grounding of $(\sentence, \domain)$ with the additional constraint $\theta$.
However, by this method, we will lose the symmetry present at the first-order level, and more importantly, such tests are generally intractable from the $\mathrm{NP}$-completeness of SAT problems.
Instead, we develop a more efficient extendability test by exploiting the symmetry of \fotwo models following the approach of \cite{meng2025model}.
We first need the notion of \emph{configuration}~\cite{beame2015symmetric}.
% We denote by $q$ the number of unary types, i.e., $q=|\utypes|$.

\begin{definition}[Configuration]
Fix an $\fotwo$ sentence $\sentence$, and let $\utypes$ be the set of valid unary types with an arbitrary fixed ordering.
Let $q=|\utypes|$.
A \emph{configuration} is a vector $\vecn=(n_1,\dots,n_q)\in \mathbb{N}^q$, where $n_i$ counts the number of elements realizing the $i$-th unary type.
% We write $|\vecn|=\sum_i n_i$ for the induced domain size.
\end{definition}

A configuration $\vecn$ is \emph{satisfiable} if there exists a model of $\sentence$, where the number of elements realizing the $i$-th unary type is exactly $n_i$ for all $i$.
It turns out that the satisfiability of configurations has a monotonicity property, which allows us to compare configurations using a partial order and characterize satisfiable configurations via finitely many bounded ones.
\begin{definition}
\label{def:derivation}
Fix an $\fotwo$ sentence $\sentence$.
For its configurations $\vecn,\vecn'\in \mathbb{N}^q$, we write $\vecn \preceq \vecn'$ iff for all $i \in [q]$,
$n_i = 0 \Rightarrow n_i'=0,$ and $n_i > 0 \Rightarrow n_i \le n_i'.$
\end{definition}
% Two key properties hold for satisfiable configurations of $\fotwo$ sentences.
% First, satisfiability is monotone w.r.t.\ $\preceq$, that is, if $\vecn$ is satisfiable and $\vecn \preceq \vecn'$, then $\vecn'$ is also satisfiable.
% Second, there exists an upper bound on the components of satisfiable configurations.
% Formally, let $m$ be the number of existential constraints $\forall x\exists y:\beta_k(x,y)$ in $\sentence$, and define
% $$
% \delta_\sentence = \max\{m(m+1),2m+1\}.
% $$
% Then for a satisfiable configuration $\vecn$ with $n_i > \delta_\sentence$,
% the configuration $\vecn[i \mapsto \delta_\sentence]$, which is identical to $\vecn$ except the $i$-th entry replaced by $\delta_\sentence$, is also satisfiable.
% Together, monotonicity and boundedness yield the following characterization.
% Then each satisfiable configuration can be derived from a bounded satisfiable configuration, as stated in the following theorem.
\begin{theorem}[{\protect\cite[Theorem~3]{meng2025model}}]
\label{th:cfg-template}
Fix an $\fotwo$ sentence $\sentence$ with $m$ existential conjuncts in Scott normal form.
A configuration $\vecn$ is satisfiable iff there exists a satisfiable configuration $\vecn'\in\{0,1,\dots,\delta\}^{q}$ such that $\vecn' \preceq \vecn$, where $\delta = \max\{m(m+1),2m+1\}$.
\end{theorem}

We note that $\delta$ and $q$ depend only on the input sentence $\sentence$.
Thus, to check whether a configuration $\vecn$ of size $n$ is satisfiable, it suffices to compare it with finitely many template configurations $\vecn'$ with bounded entries, which can be precomputed once for the fixed input sentence.
Next we show how to use this to perform extendability checks in our two-stage compilation efficiently.
% The high-level idea is that we maintain the unary type counts of the processed elements and check whether they can be extended to a satisfiable configuration of total size $n$.
% In our compilation, we maintain partial unary type counts and repeatedly ask whether they can be extended to a full satisfiable configuration. 
% By Theorem~\ref{th:cfg-template}, this reduces to checking derivability from one of finitely many template configurations, which can be precomputed once for the fixed input sentence.

\subsubsection{Checks in Stage~I}

We maintain a \emph{partial configuration} $\hat{\vecn}\in\mathbb{N}^{q}$ according to the context $\theta$, where $\hat{n}_i$ counts how many processed elements are assigned the $i$-th unary type in $\theta$.
Then we have the following lemma, whose proof is relatively straightforward and thus provided in the appendix.

\begin{lemma}
\label{lem:partial-cfg-extend}

Let $\theta$ be a context during Stage~I, and let $\hat{\vecn}$ be the partial configuration induced by $\theta$.
Then $\theta$ is extendable w.r.t.\ $(\sentence,\domain)$ iff there exists a satisfiable configuration $\vecn$ of $\sentence$ such that $|\vecn|=|\domain|$, and for all $i \in [q]$, $n_i \ge \hat{n}_i$.

% For a context $\theta$ and its induced partial configuration $\hat{\vecn}$, $\theta$ is extendable w.r.t.\ $(\sentence,\domain)$ iff there exists a satisfiable configuration $\vecn$ of $\sentence$ such that $|\vecn|=|\domain|$, and for all $i \in [q]$, $n_i \ge \hat{n}_i$.
\end{lemma}

By the above lemma, the extendability test reduces to a test on the partial configuration $\hat{\vecn}$, which can be solved by Theorem~\ref{th:cfg-template}.

\begin{proposition}
\label{prop:extendability-stage-1}
Let $\theta$ be a context during Stage I, and let $\hat{\vecn} = (\hat{n}_1, \dots, \hat{n}_{q})$ be the induced partial configuration.
Then $\theta$ is extendable w.r.t.\ $(\sentence,\domain)$ iff there exists a satisfiable configuration $\vecn'\in\{0,1,\dots,\delta\}^{q}$ such that for all $i \in [q]$, if $n'_i = 0$ then $\hat{n}_i = 0$, and $|\domain| - |\hat{\vecn}| \ge \sum_{i \in [q]} \max\{n'_i - \hat{n}_i, 0\}$.
\end{proposition}

\begin{proof}
  ($\Rightarrow$) Suppose $\theta$ is extendable w.r.t.\ $(\sentence,\domain)$.
  By Lemma~\ref{lem:partial-cfg-extend}, there exists a satisfiable configuration $\vecn$ such that $|\vecn|=|\domain|$ and for all $i$, $n_i \ge \hat{n}_i$.
  By Theorem~\ref{th:cfg-template}, there exists a satisfiable configuration $\vecn'\in\{0,1,\dots,\delta\}^{q}$ such that $\vecn' \preceq \vecn$.
  Then for all $i$, if $n'_i = 0$ then $n_i = 0$ and thus $\hat{n}_i = 0$, and if $n'_i > 0$ then $n_i \ge n'_i$ and thus $n_i \ge \max\{\hat{n}_i, n'_i\}$.
  Therefore, $|\domain| - |\hat{\vecn}| = \sum_i (n_i - \hat{n}_i) \ge \sum_i \max\{n'_i - \hat{n}_i, 0\}$.

  ($\Leftarrow$) Suppose there exists a satisfiable configuration $\vecn'\in\{0,1,\dots,\delta\}^{q}$ that meets the stated conditions.
  Let $\vecn^*$ be the configuration defined by $n^*_i = \max\{\hat{n}_i, n'_i\}$ for all $i$.
  Then for all $i$, if $n'_i = 0$ then $\hat{n}_i = 0$ and thus $n^*_i = 0$, and if $n'_i > 0$ then $n^*_i \ge n'_i$.
  Therefore, $\vecn' \preceq \vecn^*$, and by Theorem~\ref{th:cfg-template}, $\vecn^*$ is satisfiable.
  Moreover, $|\vecn^*| = \sum_i n^*_i = \sum_i \max\{\hat{n}_i, n'_i\} = |\hat{\vecn}| + \sum_i \max\{n'_i - \hat{n}_i, 0\} \le |\domain|$.
  We can extend $\vecn^*$ to a configuration $\vecn$ of size $|\domain|$ by increasing some entries, which remains satisfiable by monotonicity.
  Finally, by Lemma~\ref{lem:partial-cfg-extend}, $\theta$ is extendable w.r.t.\ $(\sentence,\domain)$.
\end{proof}

% \lnote{@Qiaolan, please check this}
% Specifically, for a partial configuration $\hat{\vecn}$, we can check whether there exists a template configuration $\bar{\vecn}$ such that $\bar{n}_i=\hat{n}_i$ for all $i$ with $\bar{n}_i\le 1$, and
% $$
% \sum_{i: \bar{n}_i>1} \max\{\bar{n}_i-\hat{n}_i,0\}\ \le\ n-\sum_{i}\hat{n}_i.
% $$
% This condition ensures that $\hat{\vecn}$ can be extended to a template configuration $\bar{\vecn}$ by increasing some of its entries, and then $\bar{\vecn}$ can be extended to a satisfiable configuration $\vecn$ with $|\vecn|=n$ by Theorem~\ref{th:cfg-template}.
% Since the finite set of template configurations depends only on $\sentence$, we precompute it once and perform the above test by iterating over all template configurations, in $O(\log n)$ time.

\subsubsection{Checks in Stage II}

We now consider extendability checks during Stage II. 
Recall that, at this point, the unary types of all elements have already been fixed in Stage I, and Stage II only branches on binary
types. 
More precisely, when processing a pair $(e_i,e_j)$ in lexicographic order, the current context $\theta$ contains the unary type assignments for all elements in $\domain$, the binary
types for all pairs involving elements in $\{e_1,\dots,e_{i-1}\}$, and the binary types for pairs $(e_i,e_\ell)$ with $\ell<j$. 
Hence, unlike Stage I, extendability is no longer determined solely
by the numbers of elements of each unary type: we must also keep track of which binary types have already been fixed and which existential requirements are still waiting for witnesses.

To encode this information, we follow~\cite{meng2025model} and reduce the extendability of a context in Stage~II to the satisfiability of a configuration for a fixed auxiliary \fotwo sentence $\sentence_{\mathrm{aux}}$. 
Concretely, for a context $\theta$ in the call $\textsc{Stage-II}((i,j),\theta)$, we consider the remaining domain $\Delta'=\{e_i,\dots,e_n\}$ and enrich it with fresh auxiliary predicates $T,Q,D,Z_1,\dots,Z_m$, where $T,Q,D$ encode which pairs have already been processed and
$Z_k$ records which elements are not yet $\beta_k$-satisfied.
This yields an induced complete configuration for the sentence
$\sentence_{\mathrm{aux}}$. 
The details of the fresh predicates, the auxiliary sentence, and the induced configuration are given in the appendix.
We rephrase the result of~\cite{meng2025model} in our terminology as follows.
% The context $\theta$ consists of a fixed unary type assignment and the ground binary types chosen for element pairs processed so far.
% Recall that we process pairs in lexicographic order of indexes, so when processing a pair $(e_i,e_j)$, the $\theta$ consists of the unary type assignments for all elements $e_1,\dots,e_n$, the binary types for all pairs involving elements $\{e_1,\dots,e_{i-1}\}$, and the binary types for pairs $(e_i,e_\ell)$ with $\ell < j$.
% % In particular, the all atoms involving elements $\{e_1, \dots, e_{i-1}\}$ are already fixed, and all binary types between $\{e_{i+1},\dots,e_n\}$ are still unassigned in $\theta$.
% This specific form of context is analogous to the \emph{substructure} considered in~\cite{meng2025model}, where they showed that the satisfiability check for such substructures (equivalent to our extendability check) can be reduced to a configuration satisfiability test of an auxiliary sentence over a smaller domain.
% We rephrase their result in our terminology as follows.
\begin{proposition}[Rephrased from {\protect\cite[Proposition~4]{meng2025model}}]
  \label{prop:substructure-extend}
  Let $\theta$ be a context during Stage II when processing the pair $(e_i,e_j)$.
  Then there exists an auxiliary \fotwo sentence $\sentence_{\mathrm{aux}}$ independent of $\domain$, such that the extendability check of $\theta$ w.r.t.\ $(\sentence,\domain)$ can be reduced to a configuration satisfiability test of $\sentence_{\mathrm{aux}}$ over the domain $\{e_i,\dots,e_n\}$.
\end{proposition}

Since $\sentence_{\mathrm{aux}}$ is independent of $n$, we can utilize the similar approach as in Stage~I to perform the configuration satisfiability test for $\sentence_{\mathrm{aux}}$.
% Thus the extendability check for $\theta$ can be done in $O(\log n)$ time.
For fixed sentence, this test involves only a constant number of arithmetic operations on counters bounded by $n$.
Thus, the extendability check for $\theta$ can be done in $O(\log n)$ time.

% \begin{theorem}
% \label{th:compilation-complexity}
%   For a fixed \fotwo sentence $\sentence$ and a domain $\domain$ of size $n$, Algorithm~\ref{alg:two-stage-algorithm} runs in time $O(N\log n)$, where $N$ is the size of the output circuit.
% \end{theorem}
% \begin{proof}
%   Straightforward from the fact that each extendability check can be done in $O(\log n)$ time, and thus the time spent on extendability checks is $O(N\log n)$.
% \end{proof}

% We note that the $O(N\log n)$ time complexity is non-trivial, even though $N$ can be exponential in $n$, since the naive method of performing extendability checks by SAT solving would generally lea
% We remark that this result does not break the exponential lower bound on the size of d-DNNF circuits for \fotwo sentences, since $N$ can be exponential in $n$.
% Conversely, if there exists a polynomial-size d-DNNF circuit for $\Phi_{\sentence,\domain}$, then our algorithm runs in polynomial time, since the time spent on extendability checks is $O(N\log n)$.
% \footnote{Conversely, a polynomial-size d-DNNF circuit does not imply a polynomial time compilation, and thus our result is non-trivial.}

\section{Compilation Optimizations}
\label{sec:optimizations}
% While the two-stage construction in Section~\ref{sec:compilation} is correct, it can be inefficient in practice because it branches over all valid unary and binary types naively. 
In this section, we present optimizations for the compilation.
The main idea is to identify and merge equivalent branches returned from recursive calls to $\textsc{Stage-I}$ and $\textsc{Stage-II}$ (recall that equivalent circuits have the same set of satisfying assignments).
For instance, let $\textsc{Or}(\mathcal{S})$ be the $\vee$-node returned from a call to $\textsc{Stage-I}(i, \theta)$.
If $\mathcal{S}$ contains two $\land$-nodes $\tau(e_i)\wedge C$ and $\tau'(e_i)\wedge C'$, such that $C$ and $C'$ obtained at Line~\ref{line:stage-i-recursion} in Algorithm~\ref{alg:two-stage-algorithm} are equivalent, then we can merge the two branches from $(\tau(e_i) \wedge C) \lor (\tau'(e_i) \wedge C')$ to $(\tau(e_i) \lor \tau'(e_i)) \wedge C$, and thus avoid the redundant recursive call and extendability check for $C'$.

% The above optimizations not only reduce the size of the resulting circuit, but also improve the efficiency of the compilation by avoiding redundant branches and extendability checks.

% For ease of presentation, we may write $\sentence\land \gamma$, where $\gamma$ is a conjunction of ground literals over $\domain$, to denote the sentence conditioned on the facts in $\gamma$.
Let us make the process formal.
We call two circuits $C$ and $C'$ returned from $\textsc{Stage-I}(i, \theta)$ (or $\textsc{Stage-II}((i, j), \theta)$) \emph{at the same level} if they are returned for the same $i$ (or the same $(i,j)$) but possibly different contexts $\theta$ and $\theta'$.
It is clear that two circuits at the same level have the same variable set.
Then we have the following proposition.

\begin{proposition}
\label{prop:equiv-circuit-sharing}
Let $C$ and $C'$ be two equivalent circuits at the same level.
Redirecting all edges to $C'$ to $C$ and removing $C'$ does not affect the correctness, determinism or decomposability of the resulting circuit.
\end{proposition}

\begin{proof}
The decomposability is preserved since $C$ and $C'$ have the same variable set.
By the equivalence of $C$ and $C'$, they have the same set of satisfying assignments.
Thus determinism is also preserved.
The correctness follows from the same reason: any satisfying assignment of the original circuit also satisfies the modified circuit, and vice versa, since $C$ and $C'$ have the same set of satisfying assignments.
\end{proof}

\begin{remark}
\cref{prop:equiv-circuit-sharing} can also apply to the compilation of structured d-DNNFs and OBDDs as in~\cref{cor:OBDD}, since the sharing is between circuits at the same level, which have the same variable set and thus are in the same vtree node for structured d-DNNFs and at the same variable for OBDDs.
\end{remark}

The pruning technique for equivalent circuits described above is not specific to our two-stage compilation, but can be applied to any d-DNNF compilation algorithm that constructs circuits by branching on different contexts.
The advantage of our two-stage compilation is that it provides clear criteria for identifying equivalent circuits from their contexts, without the need to check the equivalence of circuits explicitly, which can be costly.
In practice, we maintain a hash table indexed by an abstraction of the current context for each level.
Whenever a new recursive call is generated, we compute the context abstraction and check whether an equivalent subcircuit has already been compiled at that level.
If so, we reuse the existing subcircuit rather than constructing a new one.
This avoids redundant calls and extendability checks, and thus improves the efficiency of the compilation.
In the rest of this section, we show how to identify equivalent circuits at the same level from their contexts.

For ease of presentation, given a circuit $C$ and its corresponding context $\theta$, we may view $f\land \theta$ as a structure over $\domain$ for any assignment $f$ of the propositional variables in $C$, and thus write $f\land \theta \models \sentence$ to mean that the structure induced by $f\land \theta$ satisfies $\sentence$.

\subsection{Equivalent Circuits from Stage I}

We first consider the equivalence of circuits returned by $\textsc{Stage-I}(i+1, \theta \wedge \tau(e_i))$ at Line~\ref{line:stage-i-recursion} for different unary type choices $\tau$ but the same context $\theta$.
\begin{proposition}
\label{prop:stage1-equivalence}
Let $C$ and $C'$ be the circuits returned by $\textsc{Stage-I}(i+1, \theta \wedge \tau(e_i))$ and $\textsc{Stage-I}(i+1, \theta \wedge \tau'(e_i))$, respectively, where $\tau$ and $\tau'$ are two valid unary types for $e_i$.
If the following conditions hold:
\begin{itemize}
  \item for every $k\in[m]$, $\beta_k(x,x)\in \tau$ iff $\beta_k(x,x)\in \tau'$; and
  \item for every unary type $\tau''\in \utypes$, $\btypes(\tau,\tau'')$ is identical to $\btypes(\tau',\tau'')$,
\end{itemize}
then $C$ and $C'$ are equivalent.
\end{proposition}
Intuitively, the first condition guarantees that assigning $\tau$ or $\tau'$ to $e_i$ makes the same progress toward satisfying the existential constraints, while the second condition guarantees that they induce the same binary type compatibility with elements of any valid unary type.
Hence, the two contexts $\theta \wedge \tau(e_i)$ and $\theta \wedge \tau'(e_i)$ determine the same compilation subproblem.
\begin{proof}[Proof of \cref{prop:stage1-equivalence}]
By the compilation procedure, an assignment $f$ satisfies $C$ iff $f\land \theta \land \tau(e_i)\models \sentence$, and similarly for $C'$ and $\theta \land \tau'(e_i)$.
Therefore, to show the equivalence of $C$ and $C'$, it suffices to show that for every assignment $f$ over the propositional variables in $C$ (and $C'$), $f\land \theta \land \tau(e_i)\models \sentence$ iff $f\land \theta \land \tau'(e_i)\models \sentence$.
We only prove one direction, and the other direction follows by a symmetric argument.
Let $f$ be an assignment of $C$ such that $f\land \theta \land \tau(e_i)\models \sentence$.
We show that $f\land \theta \land \tau'(e_i)\models \sentence$.

Consider the universal constraint $\forall x\forall y:\phi(x,y)$ in $\sentence$.
Denote $\phi(x,y) \land \phi(y,x)$ by $\psi(x,y)$ for ease of presentation.
The difference between $\tau$ and $\tau'$ only affects the satisfaction of $\psi(e_i,e_i)$ and $\psi(e_i,e_\ell)$ for $\ell \neq i$, while for any other pair of elements, the satisfaction of $\psi$ is determined by $\theta$ and $f$, and thus is the same for both $\tau$ and $\tau'$.
Since $\tau'$ is valid, we have that $\tau'(e_i) \models \psi(e_i, e_i)$.
For any $e_\ell$ with $\ell \neq i$, let $\tau''$ be the unary type of $e_\ell$ in $\theta$, and let $\pi$ be the binary type for $(e_i,e_\ell)$ in $f$. 
Then by the second condition that $\pi\in \btypes(\tau,\tau'')$ iff $\pi\in \btypes(\tau',\tau'')$, we have $\tau'(e_i)\wedge \tau''(e_\ell)\wedge \pi(e_i,e_\ell) \models \psi(e_i,e_\ell)$.
Therefore, $f\land \theta \land \tau'(e_i)\models \forall x\forall y:\phi(x,y)$.

The satisfaction of the existential constraints $\forall x\exists y:\beta_k(x,y)$ for $k\in[m]$ can be analyzed similarly.
We only need to show that $f\land \theta \land \tau'(e_i)\models \bigwedge_{k\in[m]} \exists y:\beta_k(e_i,y)$.
For any $k\in [m], \tau(e_i)$ and $\tau'(e_i)$ interpret the same $\beta_k(e_i, e_i)$ by the first condition, and therefore, if $\exists y:\beta_k(e_i,y)$ is satisfied in $\tau(e_i)$ by $\beta_k(e_i,e_i)$, then it is also satisfied in $\tau'(e_i)$.
Otherwise, there must be some $e_\ell$ with $\ell \neq i$ such that $\beta_k(e_i,e_\ell)$ holds in $f\land \theta$, and thus also holds in $f\land \theta \land \tau'(e_i)$. 
Therefore, $f\land \theta \land \tau'(e_i)\models \bigwedge_{k\in[m]} \exists y:\beta_k(e_i,y)$.
\end{proof}

% In practice, the abstraction used as the hash key consists of three components:
% the satisfaction state $S_\theta$;
% the family of sets $\btypes(\tau_\ell,\tau_t)$ for pairs of processed elements;
% and the family of sets $\btypes(\tau_\ell,\tau)$ for each processed element $e_\ell$ and each valid unary type $\tau\in \utypes$.
% \cref{prop:extendability-stage-1} shows that this abstract signature is sufficient for detecting equivalent residual subproblems at Stage~I.

\subsection{Equivalent Circuits from Stage II}

Next, we consider the circuits returned by $\textsc{Stage-II}((i, j), \theta)$ for different contexts $\theta$ and $\theta'$ but the same $(i,j)$.
% The conditions for equivalence are more involved than in Stage I because $\theta$ and $\theta'$ may differ in the binary types for pairs of processed elements, which affects the satisfaction of existential constraints and satisfying assignments of the resulting circuits.
% For instance, when processing the pair $(e_i,e_j)$, the contexts $\theta$ and $\theta'$ may differ in the binary types for pairs involving $e_i$ or $e_j$, which might lead to different satisfaction of the existential constraints $\exists y:\beta_k(e_i,y)$ and $\exists y:\beta_k(e_j,y)$, and thus inequivalent circuits.
Although these two contexts correspond to the same stage of the compilation, they may differ in the unary types assigned to elements and in the binary types assigned to processed pairs.
Such differences affect the subproblem primarily through (i) the progress already made toward satisfying the existential constraints,  and (ii) the binary type compatibility induced for the pairs that remain to be processed.

To capture this progress, we maintain a \emph{satisfaction state} $S_\theta$ for each context $\theta$.
Formally, the satisfaction state $S_\theta$ of $\theta$ is an $n\times m$ Boolean matrix, where for each $\ell\in [n]$ and $k\in[m]$, 
$$
S_\theta(\ell,k)=\mathrm{True}
\quad\Leftrightarrow\quad
\exists \ e_t\in \domain\ \text{s.t.}\ \theta \models \beta_k(e_\ell,e_t).
$$
Moreover, by the time Stage~II starts, all unary type choices have already been fixed within each context, although not necessarily identically across different contexts. 
Hence, to characterize the residual compilation subproblem at $\textsc{Stage-II}((i,j),\theta)$, it suffices to record the current satisfaction state together with the binary type compatibility for the pairs whose binary types have not yet been assigned. 
This yields the following proposition.
\begin{proposition}
\label{prop:stage2-equivalence}
Fix a pair $(i,j)$. Let $C$ and $C'$ be the circuits returned by $\textsc{Stage-II}((i, j), \theta)$ and $\textsc{Stage-II}((i, j), \theta')$, respectively.
If the following conditions hold:
\begin{itemize}
  \item $S_{\theta}=S_{\theta'}$; and
  \item for any unprocessed pair of elements $(e_\ell,e_t)$, $\btypes(\tau_\ell,\tau_t)$ is identical to $\btypes(\tau_\ell', \tau_t')$, where $\tau_\ell$ and $\tau_t$ (resp. $\tau_\ell'$ and $\tau_t'$) are the unary types of $e_\ell$ and $e_t$ in $\theta$ (resp. $\theta'$),
\end{itemize}
then $C$ and $C'$ are equivalent.
\end{proposition}
The proof is similar to that of~\cref{prop:stage1-equivalence}, and therefore provided in the appendix.
Note that for a sentence without existential constraints, the first condition in \cref{prop:stage2-equivalence} is vacuous. 
Hence, the equivalence of subcircuits returned by Stage II depends only on the valid binary type sets for the unprocessed pairs.

\begin{example}
  Consider the sentence $\sentence_{RB}$ in \cref{ex:2-color}, which has no existential constraints, over the vertex domain $\{e_1,e_2,e_3\}$.
  Let
  $
    \theta = \tau_R(e_1) \wedge \tau_B(e_2) \wedge \tau_B(e_3)
  $ 
  and
  $
    \theta' = \tau_B(e_1) \wedge \tau_R(e_2) \wedge \tau_R(e_3)
  $
  be two contexts obtained after Stage I, where $\tau_R$ and $\tau_B$ are the unary types for red and blue, respectively.
  Suppose Stage II is about to process the first pair $(e_1,e_2)$.
  Since $\sentence_{RB}$ has no existential constraints, we have $S_\theta = S_{\theta'} = \emptyset$.
  Moreover, by \cref{ex:2-color}, for each unprocessed pair, the valid binary type set under $\theta$ is identical to that under $\theta'$.
  By Proposition~\ref{prop:stage2-equivalence}, the circuits returned by $\textsc{Stage-II}((1,2),\theta)$ and $\textsc{Stage-II}((1,2),\theta')$ are equivalent.
  Hence, once one of these subcircuits has been computed, it can be reused for the other without making the second recursive call.
\end{example}

If existential constraints are present, then the satisfaction state must also be taken into account when checking the equivalence of subcircuits returned by Stage II.

\begin{example}
  Consider the sentence $\sentence_E$ in \cref{ex:context} over the vertex domain $\{e_1,e_2,e_3,e_4\}$.
  Since $\sentence_E$ has a unique valid unary type, the second condition in Proposition~\ref{prop:stage2-equivalence} is automatically satisfied.
  Suppose Stage II is about to process the pair $(e_2,e_4)$, and let
  $$
    \theta
    = \pi_E(e_1,e_2) \wedge \pi_{\neg E}(e_1,e_3) \wedge \pi_{\neg E}(e_1,e_4) \wedge \pi_E(e_2,e_3)
  $$
  and
  $$
    \theta'
    = \pi_E(e_1,e_2) \wedge \pi_E(e_1,e_3) \wedge \pi_{\neg E}(e_1,e_4) \wedge \pi_{\neg E}(e_2,e_3),
  $$
  where we omit the common unary type assignments.
  The two different contexts induce the same satisfaction state: in both $\theta$ and $\theta'$, the elements $e_1$, $e_2$, and $e_3$ already satisfy the existential requirement $\exists y\, E(x,y)$, while $e_4$ still has no witness.
  Therefore, by Proposition~\ref{prop:stage2-equivalence}, the circuits returned by $\textsc{Stage-II}((2,4),\theta)$ and $\textsc{Stage-II}((2,4),\theta')$ are equivalent.
\end{example}

Note that the above optimizations incur at most quadratic overhead in the domain size. 
Hence, the optimized compilation satisfies the following time complexity guarantee.
\begin{theorem}
\label{thm:optimized-compilation-complexity}
For a fixed \fotwo sentence $\sentence$ and a finite domain $\domain$, \cref{alg:two-stage-algorithm} equipped with the optimizations runs in time
$O\big((|C|+1)\, |\domain|^2\big)$, where $|C|$ is the size of the output circuit produced by
the optimized algorithm.
\end{theorem}

\begin{proof}[Proof sketch]
We first analyze the time spent on $\textsc{Stage-I}$ and $\textsc{Stage-II}$ without considering the time spent on recursive calls.
Note that the sentence is fixed.
Every $\textsc{Stage-I}$ (or $\textsc{Stage-II}$) traverses a constant number of unary (or binary) types, and hence performs a constant number of $\textsc{Extendable}$ tests and hash table lookups for optimizations.
From \cref{prop:extendability-stage-1} and \cref{prop:substructure-extend}, we know that any $\textsc{Extendable}$ test takes $O(\log |\domain|)$ time.
Furthermore, the Stage~I memoization data have constant size due to the constant number of unary types, and the Stage~II abstraction has size $O(|\domain|^2)$ (an $O(|\domain|)$ satisfaction state and $O(|\domain|^2)$ compatibility data on unprocessed pairs).
Therefore, the time spent on any call to $\textsc{Stage-I}$ and $\textsc{Stage-II}$ without considering recursive calls is $O(|\domain|^2)$.

It remains to bound the number of recursive calls. 
Every non-root recursive call is generated only from a branch, where the $\textsc{Extendable}$ test succeeds. 
Each such successful branch contributes a distinct edge to the output circuit, even if the output circuit is reused by multiple branches.
Hence the total number of recursive calls is $O(|C|+1)$. 
Therefore the total running time is $O((|C|+1)\, |\domain|^2)$. 
The full proof can be found in the appendix.
\end{proof}

% \begin{figure}[t]
% \centering
% \includegraphics[width=0.6\textwidth]{fig/ratio_scatter_n10.pdf}
% \caption{}
% \end{figure}

\section{Evaluation}
\label{sec:experiments}

We implemented our two-stage compilation algorithm and evaluated it against several propositional d-DNNF compilers.
All experiments were conducted on a machine with an Intel Xeon 8383C CPU and 128\,GB of RAM.

We considered several representative benchmark sentences that have been widely used in the literature, including $\sentence_{RB}$ from \cref{ex:2-color} (2-colored graphs), $\sentence_E$ from \cref{ex:context} (graphs without isolated vertices), $\sentence_{RBE} = \sentence_{RB} \wedge \sentence_{E}$, $\sentence_P$ from the proof of \cref{th:hardness_fo2} (permutation-like structures), and $\sentence_{D}$ which defines dominating sets:
\begin{align*}
\! \sentence_D = \big(\forall x: \neg E(x,x)\big) \wedge \big(\forall x\forall y:E(x,y)\rightarrow E(y,x)\big) \wedge \big(\forall x\exists y: \neg D(x) \rightarrow E(x,y)\wedge D(y)\big).
\end{align*}
We also evaluated a parametric benchmark family \texttt{ui-bj}:
\begin{align*}
\texttt{ui-bj} = &\forall x: \text{ExactlyOne}\big[ C_1(x), C_2(x), ..., C_i(x)\big] \ \wedge \\
&\bigwedge_{k\in[j]} \forall x\forall y: \big(E_k(x,y) \rightarrow E_k(y,x)\big) \wedge \bigwedge_{\substack{ k' \in [j] : \\ k'\neq k}} \forall x\forall y: \big(E_k(x,y) \rightarrow \neg E_{k'}(x,y)\big) \ \wedge \\
&\forall x\forall y: \big(\bigvee_{k\in[j]} E_k(x,y)\big) \rightarrow \bigwedge_{\ell\in[i]} \big(\neg C_\ell(x) \vee \neg C_\ell(y)\big) \ \wedge \ \bigwedge_{k\in[j]} \forall x\exists y: E_k(x,y).
\end{align*}
Here, $\text{ExactlyOne}$ abbreviates the constraint that exactly one of $C_1(x),\dots,C_i(x)$ holds.
We evaluated the settings $(i,j)\in\{(4,2), (2,2), (6,2), (4,1), (4,3)\}$.
In this family, each binary predicate $E_k$ is symmetric, and different predicates are mutually exclusive.
To further study the effect of structural properties, we also considered several variants of \texttt{u4-b2} with different interaction patterns between the binary predicates; full details are given in the appendix.

\begin{table}[t]
\centering
\small
\caption{Circuit-size / compilation-time ratios relative to our method. For each sentence, we report the largest domain size for which all compared methods successfully produced a d-DNNF circuit within the time limit. The best results are shown in bold and the second-best results are underlined. For Bella, we tested three hypergraph partitioning methods (PaToH, KaHyPar, and Cara); we report the size and runtime of the setting that produced the smallest circuit.}
\begin{tabular}{c|cccccccc}
\toprule

$(\sentence, |\domain|)$ & ours & bella & d4 & c2d \\
\midrule
$(\sentence_{RB}, 20)$ & \textbf{1.00}/\textbf{1.00} & 1.55/\underline{1.30} & 1.50/1.45 & \underline{1.23}/1.56 \\

$(\sentence_E, 15)$ & \textbf{1.00}/\textbf{1.00} & 10.08/27.97 & \underline{7.80}/\underline{12.71} & 86.91/312.88 \\

$(\sentence_{RBE}, 12)$ & \textbf{1.00}/\textbf{1.00} & \underline{1.12}/\underline{2.87} & 4.74/16.01 & 3.07/13.44 \\

$(\sentence_P, 10)$ & \underline{1.00}/\textbf{1.00} & 1.39/3.29 & \textbf{0.49}/\underline{1.10} & 12.65/47.84 \\

$(\sentence_D, 10)$ & \textbf{1.00}/\textbf{1.00} & 1.86/24.48 & \underline{1.34}/\underline{9.19} & 29.29/1431.64 \\
\midrule
$(\texttt{u2-b2}, 10)$ & \textbf{1.00}/\textbf{1.00} & 5.44/8.59 & 7.17/8.66 & \underline{2.57}/\underline{5.09} \\

$(\texttt{u4-b1}, 8)$ & \textbf{1.00}/\textbf{1.00} & 17.47/169.14 & 13.96/\underline{45.83} & \underline{7.39}/125.95 \\

$(\texttt{u4-b2}, 6)$ & \textbf{1.00}/\textbf{1.00} & 14.39/26.31 & 16.33/\underline{10.58} & \underline{4.35}/14.51 \\

$(\texttt{u4-b3}, 6)$ & \textbf{1.00}/\textbf{1.00} & 30.62/37.11 & 70.43/37.54 & \underline{4.10}/\underline{9.34} \\

$(\texttt{u6-b2}, 6)$ & \textbf{1.00}/\textbf{1.00} & \underline{142.66}/\underline{136.23} & 474.05/171.70 & 336.07/434.61 \\
\midrule
$(\texttt{u4-b2-ad}, 6)$ & \textbf{1.00}/\textbf{1.00} & 26.34/38.70 & 51.23/56.43 & \underline{4.79}/\underline{12.86} \\

$(\texttt{u4-b2-so}, 9)$ & \textbf{1.00}/\textbf{1.00} & 3.69/9.65 & 2.67/\underline{4.73} & \underline{2.22}/13.94 \\

$(\texttt{u4-b2-cc}, 6)$ & \textbf{1.00}/\textbf{1.00} & 9.19/16.99 & 12.78/\underline{9.63} & \underline{4.84}/16.68 \\

$(\texttt{u4-b2-fd}, 5)$ & \textbf{1.00}/\textbf{1.00} & 35.81/99.07 & 37.29/47.15 & \underline{4.14}/\underline{8.83} \\

$(\texttt{u4-b2-ao}, 8)$ & \textbf{1.00}/\textbf{1.00} & 5.38/54.77 & 7.22/7.93 & \underline{1.40}/\underline{4.68} \\

$(\texttt{u4-b2-fo}, 10)$ & 1.00/1.00 & \underline{0.04}/1.13 & 0.07/\textbf{0.45} & \textbf{0.02}/\underline{0.71} \\

\bottomrule
\end{tabular}
\label{tab:res}
\end{table}

We validated the correctness of our implementation by cross-checking model enumeration and model counts against propositional baselines.
For performance evaluation, we compared against three propositional d-DNNF compilers: Bella~\cite{bella}, d4~\cite{d4}, and c2d~\cite{c2d}.
Given a sentence $\sentence$ and a domain $\domain$, we first constructed the propositional grounding of $(\sentence,\domain)$, and then invoked each baseline compiler using its recommended or default settings.
We increased the domain size starting from $n=4$ until compilation exceeded the timeout of $1800$ seconds, or up to $n=20$.
For randomized methods, we report averages over $10$ successful runs. For deterministic and stable methods, we report averages over 3 runs. All experiments were executed single-threadedly.
The details of this pipeline are given in the appendix.

\cref{tab:res} reports the circuit-size and compilation-time ratios relative to our method. 
Each row corresponds to the largest domain size for which all compared methods successfully finished within the time limit. 
Overall, our method produces the smallest circuits in most of these representative comparisons, and it is also the fastest in many cases.
The main exceptions are $\sentence_P$ and \texttt{u4-b2-fo}, where our method is less effective. 
In both cases, the binary atoms are essentially unconstrained and can be handled more independently. 
Our current implementation does not exploit this independence: it still constructs binary types for $E(x,y)$ and $E(y,x)$ even when they could be processed separately.
In addition, we also evaluated the OBDD compilation and the full results are given in the appendix.

The empirical gains of our compiler come from three main sources. 
First, the compiler branches on types rather than on individual ground atoms. 
A type decision fixes a whole block of literals associated with an element or a pair, and therefore yields a significant reduction in the search space and more compact circuit than branching on atoms. 
Second, efficient extendability checks prune unsatisfying branches and hence avoid unnecessary recursive calls and circuit construction.
This is particularly useful for sentences with existential requirements.
% where a partial type assignment may already make it impossible for some element to obtain a required witness. 
Third, the subcircuit merging optimization allows the compiler to reuse previously computed subcircuits for different contexts that induce the same compilation subproblem.

\section{Conclusion and Discussion}
We studied the compilation of \fotwo{} sentences over finite domains into standard DNNF-based target languages. 
On the negative side, we showed that there exists a fixed \fotwo{} sentence whose grounding over a domain of size $n$ requires DNNF size $2^{\Omega(n)}$. 
On the positive side, we presented a specialized compilation method by exploiting some characterization of \fotwo{}.
% On the positive side, we developed a two-stage d-DNNF compiler that exploits the type structure of \fotwo{} sentences, branches on types rather than on individual ground atoms. 
% We further incorporated the optimizations that merge equivalent subcircuits and reduce redundant recursive calls and extendability checks.
Our experiments show that this approach is effective on several representative sentences, improving over baseline propositional compilation pipelines in both circuit size and compilation time. 
% Taken together, these results clarify the compilation landscape for \fotwo{}: compact compilation is impossible in the worst case, yet its restricted structure still enables a principled and effective specialized compiler.
Finally, we discuss some open questions and future directions.

It is well-known that the order of variable branching can significantly affect the performance of propositional compilers, such as c2d and cudd.
In our proposed method, we use a restricted processing order. 
It first branches unary types and then binary types. 
We remark that this restriction is not merely an implementation choice but is crucial for the efficiency of the method, more specifically for the efficiency of extendability checks.
Interleaving unary and binary type decisions would break the symmetry and make the extendability checks even fall back to the \#P-hardness of \fotwo{} satisfiability with arbitrary binary evidence~\cite{van_den_broeck_conditioning_2012-1}.
We however note that recent work on \fotwo{} model counting with binary evidence that is of bounded treewidth~\cite{kuula2026tractable} may provide a promising direction for extending the method to handle more flexible processing orders, that we leave for future work.

Another natural question is whether our method can be extended to $\mathbf{FO}^3\, $?
Such an extension is nontrivial. 
The feature of two variables is used in an essential way throughout our method. 
First, models of \fotwo{} can be organized by unary types and binary types, which makes the two-stage decomposition possible. 
For $\mathbf{FO}^3$, constraints relate more elements simultaneously, so pairwise type compatibility is no longer sufficient. 
Second, our extendability checks rely crucially on the bounded configuration characterization of~\cref{th:cfg-template}.  
This is closely related to the classical boundary between the well-behaved two-variable fragment and richer fragments with three or more variables.
Therefore, instead of targeting unrestricted $\mathbf{FO}^3$ directly, a promising direction is to study extensions of $\mathbf{FO}^2$ such as counting extensions like \ctwo or \fotwo{} with additional restricted axioms, e.g., tree or linear order axioms.

The last potential direction is to combine the type-level approach of our method with finer-grained propositional compilation techniques, which is suggested by our experiments on $\sentence_P$ and \texttt{u4-b2-fo}.
The current compiler operates primarily at the level of types, and may therefore miss finer-grained propositional independence in some instances.
%%
%% Bibliography
%%

%% Please use bibtex, 

\bibliography{lipics-v2021-sample-article}

\appendix

\section{Proof of Lemma~\ref{lem:partial-cfg-extend}}

\begin{proof} % [Proof of Lemma~\ref{lem:partial-cfg-extend}]
  If $\theta$ is extendable w.r.t.\ $(\sentence,\domain)$, then it has an extension $\theta'$ over $\domain$ that is consistent with some model of $\sentence$.
  Let $\vecn$ be the configuration induced by $\theta'$.
  We have that $|\vecn|=|\domain|$ and $n_i\ge \hat{n}_i$ for all $i$, since $\theta'$ extends the unary-type choices already made in $\theta$. Conversely, given a satisfiable configuration $\vecn$ with $|\vecn|=|\domain|$ and $n_i\ge \hat{n}_i$, we can extend $\theta$ to a unary-type assignment $\theta'$ on $\domain$ whose induced configuration is exactly $\vecn$ by assigning the remaining elements to match the required counts.
  Satisfiability of $\vecn$ then provides a model of $\sentence$, which can be renamed to domain $\domain$ to obtain a model consistent with $\theta'$, and hence $\theta$ is extendable.
\end{proof}

\section{The Details of Extendability Checks in Stage II}
\label{app:stage2-checks}

In this section, we make explicit the auxiliary sentence used in \cref{prop:substructure-extend} and the corresponding extendability checks in Stage II.

Recall that the input sentence is in Scott normal form:
\begin{align*}
\sentence \;=\; \forall x \forall y: \, \phi(x,y)
\;\wedge\;
\bigwedge_{k\in[m]} \forall x \exists y: \, \psi_k(x,y),
\end{align*}
Consider a context $\theta$ arising in Stage II after a binary type has been fixed for the current pair $(e_i,e_j)$, i.e., $\theta$ contains:
\begin{itemize}
  \item the unary type assignments for all elements $e_1,\dots,e_n$;
  \item the binary types for all pairs involving some element in $\{e_1,\dots,e_{i-1}\}$; and
  \item the binary types for the pairs $(e_i,e_\ell)$ with $i < \ell \le j$.
\end{itemize}
Let $\domain' = \{e_i,\dots,e_n\}$.
Since all unary types are already fixed in Stage I and every processed binary type is chosen from the valid set $\btypes(\tau,\tau')$, the universal conjunct $\forall x \forall y : \phi(x,y)$
has already been respected on the processed part. 
Thus, for future extendability, the remaining
information we need to record is:
\begin{enumerate}
  \item which pairs in $\domain'$ have already been fixed by the current context; and
  \item which elements are still waiting for witnesses for the existential conjuncts.
\end{enumerate}

To encode this information, we introduce fresh predicates $T/1$, $Q/1$, $D/2$, and $Z_k/1$ for each $k \in [m]$. 
Their intended interpretations over $\domain'$ are as follows.
\begin{itemize}
  \item $T$ marks the first element of $\domain'$, namely $e_i$:
  $$
  T(e_i) \text{ is true, and } T(e_\ell) \text{ is false for all } \ell > i.
  $$

  \item $Q$ marks those elements $e_\ell$ such that the binary type of $(e_i,e_\ell)$ has
 been fixed in $\theta$:
  $$
  Q(e_\ell) \text{ is true for } \ell \in \{i,\dots,j\},
  \text{ and }
  Q(e_\ell) \text{ is false for } \ell > j.
  $$

  \item $D$ marks ordered pairs whose binary types have already been fixed by $\theta$ inside
  $\domain'$:
  $$
  D(e_i,e_\ell) \text{ and } D(e_\ell,e_i) \text{ are true for every } \ell \in \{i,\dots,j\},
  $$
  and $D(a,b)$ is false for all other ordered pairs $(a,b) \in \domain' \times \domain'$.

  \item For each $k \in [m]$, the predicate $Z_k$ marks the elements that are not yet
  $\beta_k$-satisfied in $\theta$:
  $$
  Z_k(e) \text{ is true iff } \theta \not\models \exists y : \beta_k(e,y).
  $$
\end{itemize}
Intuitively, $T,Q,D$ describe which pairs are already unavailable for future choices, while
$Z_k$ records the remaining witness obligations.

We now define an auxiliary \fotwo sentence $\sentence_{\mathrm{aux}}$ over the expanded
vocabulary:
\begin{align}
\sentence_{\mathrm{aux}} \;=\;
&\ \forall x \forall y : \phi(x,y) \ \wedge \label{eq:app-aux0} \\
&\ \forall x \forall y :
   T(x) \wedge Q(y) \rightarrow D(x,y) \wedge D(y,x) \ \wedge \label{eq:app-aux1} \\
&\ \forall x \forall y :
   D(x,y) \rightarrow \bigl((T(x)\wedge Q(y)) \vee (T(y)\wedge Q(x))\bigr) \ \wedge \label{eq:app-aux2} \\
&\ \bigwedge_{k \in [m]}
   \forall x : Z_k(x) \rightarrow \exists y : \beta_k(x,y) \wedge \neg D(x,y). \label{eq:app-aux3}
\end{align}

The meaning of these conjuncts is as follows.
\begin{itemize}
  \item \eqref{eq:app-aux0} is simply the universal part inherited from $\sentence$.
  \item \eqref{eq:app-aux1} says that every pair consisting of the distinguished element
  marked by $T$ and an element already marked by $Q$ is recorded by $D$.
  \item \eqref{eq:app-aux2} gives the converse direction: if $D(x,y)$ holds, then one of
  $x,y$ must be the distinguished element and the other must already lie in the processed
  prefix of the current row. Hence, \eqref{eq:app-aux1} and \eqref{eq:app-aux2} together ensure that $D$ marks exactly those pairs in $\domain'$ whose binary types have already been fixed by the current context.
  \item Finally, \eqref{eq:app-aux3} says that if an element $x$ is still missing a witness
  for the $k$-th existential conjunct, then such a witness must be found using a pair
  $(x,y)$ whose binary type has not yet been fixed, i.e., for which $\neg D(x,y)$ holds.
\end{itemize}

The context $\theta$ already fixes the unary type of every element in $\Delta'$. Together with
the truth values of the auxiliary unary predicates $T,Q,Z_1,\dots,Z_m$ defined above, this
determines a complete configuration for the fixed sentence $\sentence_{\mathrm{aux}}$ over
$\domain'$. We denote this configuration by $\mathbf{n}^{\mathrm{aux}}_\theta$.

We emphasize that the concrete binary types already chosen in $\theta$ need not be encoded
again in the configuration: they have already been verified to be valid with respect to the
universal conjunct, and for the purpose of future extendability, what matters is only which
pairs are still available for choosing new witnesses and which witness obligations remain open.
These are precisely the roles of $D$ and the predicates $Z_k$.

We now justify the reduction used in Section~5.3.2.
\begin{proof}[Proof of Proposition~15]
We show that $\theta$ is extendable w.r.t.\ $(\sentence,\domain)$ if and only if the induced
configuration $\mathbf{n}^{\mathrm{aux}}_\theta$ is satisfiable for $\sentence_{\mathrm{aux}}$ over $\domain'$.

\medskip
($\Rightarrow$)
Assume that $\theta$ is extendable. Then there exists a model $\structure$ over $\domain$
such that  $\structure \models \sentence \wedge \theta$.
Restrict $\structure$ to the domain $\domain'$ and expand it by interpreting the auxiliary
predicates $T,Q,D,Z_k$ exactly as defined above. 
We claim that the resulting structure $\structure'$ satisfies $\sentence_{\mathrm{aux}}$.

First, \eqref{eq:app-aux0} holds because $\structure$ is a model of $\sentence$, and hence
satisfies the universal conjunct $\forall x \forall y : \phi(x,y)$.
Next, \eqref{eq:app-aux1} and \eqref{eq:app-aux2} hold by construction of $T,Q,D$:
the relation $D$ was defined precisely to mark those pairs in $\domain'$ whose binary types
have already been fixed by $\theta$.
Finally, consider any $k \in [m]$ and any element $e_t \in \domain'$ such that $Z_k(e_t)$ holds.
By definition of $Z_k$, the context $\theta$ does not yet guarantee a witness for
$\beta_k(e_t,y)$. 
Since $\structure \models \forall x \exists y : \beta_k(x,y)$, there must exist some $e_\ell \in \domain$ such that $\structure \models \beta_k(e_t,e_\ell)$. 
Moreover, if the pair $(e_t,e_\ell)$ had already been fixed by $\theta$ inside $\domain'$, then $e_t$ would already be $\beta_k$-satisfied in $\theta$, contradicting $Z_k(e_t)$. Therefore, when $e_\ell \in \domain'$, we must have $\neg D(e_t,e_\ell)$. 
This proves \eqref{eq:app-aux3}. 
Hence
$
\structure' \models \sentence_{\mathrm{aux}}.
$
By construction, the complete unary information of $\structure'$ over the expanded vocabulary
is exactly $\mathbf{n}^{\mathrm{aux}}_\theta$, so this configuration is satisfiable.

\medskip
($\Leftarrow$)
Conversely, assume that the configuration $\mathbf{n}^{\mathrm{aux}}_\theta$ is satisfiable for $\sentence_{\mathrm{aux}}$ over $\domain'$. 
Then there exists a model $\structure'$ of $\sentence_{\mathrm{aux}}$ over $\domain'$ realizing exactly this configuration.

Because $\mathbf{n}^{\mathrm{aux}}_\theta$ was induced from the current Stage II context,
the structure $\structure'$ agrees with $\theta$ on the unary types of all elements in
$\domain'$, on which element is marked by $T$, on which elements are marked by $Q$, and on
which elements are still missing witnesses for each existential conjunct. 
By \eqref{eq:app-aux1} and \eqref{eq:app-aux2}, the predicate $D$ in $\structure'$ marks exactly those pairs in $\domain'$ whose binary types have already been fixed by $\theta$.

Now keep all type choices that are already fixed in $\theta$. 
For the remaining unprocessed pairs, choose binary types so as to realize the existential witnesses guaranteed by \eqref{eq:app-aux3}. 
This is possible because \eqref{eq:app-aux3} only asks for witnesses on pairs outside $D$, i.e., on pairs whose binary types are still free to choose. 
Since all fixed choices in $\theta$ are valid and the newly chosen pairs satisfy the required $\beta_k$-atoms,
the resulting completion extends $\theta$ to a full model of $\sentence$ over $\domain$.
Therefore, $\theta$ is extendable w.r.t.\ $(\sentence,\domain)$.

\medskip
This equivalence is exactly the content of~\cite[Proposition~4]{meng2025model}, rewritten in
the notation of the present paper.
\end{proof}

\begin{remark}
Once the sentence $\sentence$ is fixed, the auxiliary sentence $\sentence_{\mathrm{aux}}$ is also fixed. 
Hence, by \cref{th:cfg-template}, the set of satisfiable bounded configurations of $\sentence_{\mathrm{aux}}$ can be precomputed once and for all. 
The extendability check in Stage II is therefore reduced to testing whether the induced configuration $\mathbf{n}^{\mathrm{aux}}_\theta$ can be derived from one of these finitely many template configurations. 
\end{remark}

\section{Proof of~\cref{prop:stage2-equivalence}}
\begin{proof}
  
The argument is similar to the proof of \cref{prop:stage1-equivalence}, where we also discuss the universal and existential constraints separately.
Let $f$ be an assignment satisfying $f\land \theta\models \sentence$.
We show that $f\land \theta' \models \sentence$.

Consider the universal constraint $\forall x\forall y:\phi(x,y)$.
Recall $\psi(x,y) := \phi(x,y)\land \phi(y,x)$.
Since all unary types in $\theta'$ are valid, we have $\psi(e_\ell, e_\ell)$ holds in $\theta'$ for every $\ell\in [n]$.
Let $e_\ell$ and $e_t$ be two elements in $\domain$.
If $(\ell,t) < (i,j)$, $\theta'$ already fixes the binary type for $(e_\ell,e_t)$ that has passed the extendability check, so $\psi(e_\ell,e_t)$ holds in $\theta'$ and thus holds in $f\land \theta'$ as well.
Otherwise, $(\ell,t) > (i,j)$ and the binary type for $(e_\ell,e_t)$ is not yet fixed in both $\theta$ and $\theta'$.
Let $\pi$ be the binary type for $(e_\ell,e_t)$ in $f$. 
By the second condition, we have $\pi\in \btypes(\tau_\ell',\tau_t')$, and thus $\tau_\ell'\wedge \tau_t'\wedge \pi(e_\ell,e_t) \models \psi(e_\ell,e_t)$. 
Putting all cases together, we have $f\land \theta' \models \forall x\forall y:\phi(x,y)$.

Next, consider the existential constraints $\bigwedge_{k\in[m]} \forall x\exists y:\beta_k(x,y)$.
For any $k\in[m]$, since $f\land\theta \models \forall x\exists y:\beta_k(x,y)$, we have that for any $\ell \in [n]$, there must be an element $e_t\in \domain$ such that $\beta_k(e_\ell,e_t)$ holds in $f\land \theta$. There are two cases:
\begin{itemize}
  \item If $\beta_k(e_\ell,e_t)$ holds in $\theta$, then $S_{\theta}(\ell,k)$ is $\mathrm{True}$, and so is $S_{\theta'}(\ell,k)$ by the equivalence $S_{\theta}=S_{\theta'}$. By the definition of satisfaction state, there must exist some $e_{t'}\in \domain$ such that $\beta_k(e_\ell,e_{t'})$ holds in $\theta'$, and thus holds in $f\land \theta'$ as well.
  \item Otherwise, $\beta_k(e_\ell,e_t)$ must hold in the assignment $f$, and thus also holds in $f\land \theta'$.
\end{itemize}
Therefore, $f\land \theta' \models \bigwedge_{k\in[m]} \forall x\exists y:\beta_k(x,y)$.
\end{proof}

\section{Proof of \cref{thm:optimized-compilation-complexity}}

\begin{proof}
We analyze the \cref{alg:two-stage-algorithm} equipped with the optimizations in~\cref{sec:optimizations}.
For each level, we maintain a hash table indexed by the abstraction of the context. 
Whenever a recursive call is generated, we compute its abstraction and check whether a subcircuit with the same abstraction has already been compiled at that level. 
If so, we reuse the existing subcircuit; otherwise, we construct a new one. 
By \cref{prop:stage1-equivalence,prop:stage2-equivalence}, equality of the corresponding abstractions is a sufficient condition for such reuse. 
Hence, a cache miss always creates a new subcircuit root at that level.

We say that a recursive call to Stage~I or Stage~II is \emph{expanded} if it is a cache miss and its body is actually executed, i.e., the call is not immediately answered by reusing an already compiled subcircuit. 
We first bound the cost of one expanded call, and then bound the total number of expanded calls by the size of the output circuit.

For Stage~I, the loop ranges over the valid unary types $\utypes$.
For Stage~II, the loop ranges over the valid binary types $\btypes(\tau_i,\tau_j)$ for the current pair. 
Since $\sentence$ is fixed, both $|\utypes|$ and $\max_{\tau,\tau'} |\btypes(\tau,\tau')|$ are constants depending only on $\sentence$, and each type contains only $O(1)$ literals. 
Hence each expanded call considers only constantly many candidate branches.

We now bound the \emph{exclusive} work performed by a single expanded call, i.e., the work done in the body of the call excluding the recursive work of newly expanded descendants.

For each candidate branch, the algorithm performs one $\textsc{Extendable}$ test. 
By~\cref{prop:extendability-stage-1}, the extendability test in Stage I can be carried out by comparing the current partial configuration with finitely many precomputed bounded templates; by~\cref{prop:substructure-extend}, the extendability test in Stage II reduces to satisfiability of a configuration for a fixed auxiliary sentence. 
In both cases, since $\sentence$ is fixed, each $\textsc{Extendable}$ test runs in $O(\log |\domain|)$ time.

The memoization overhead is $O(1)$ per candidate in Stage~I and $O(|\domain|^2)$ per candidate in Stage~II. 
Indeed, in Stage~I the abstraction only records the data appearing in~\cref{prop:stage1-equivalence}: the truth values of the diagonal existential atoms $\beta_k(x,x)$ and the binary type compatibility pattern with each unary type in $\utypes$. 
Since both $m$ and $|\utypes|$ depend only on $\sentence$, this abstraction has constant size. 
In Stage~II, the abstraction records:
\begin{enumerate}
    \item the satisfaction state $S_\theta$, which is an $|\domain| \times m$ Boolean matrix and hence has size $O(|\domain|)$, and
    \item the binary type compatibility on all unprocessed pairs, which involves $O(|\domain|^2)$ pairs and therefore has size $O(|\domain|^2)$.
\end{enumerate}
Consequently, computing, hashing, comparing, and inserting a Stage~II abstraction can all be done in $O(|\domain|^2)$ time. 
Local circuit construction at a branch adds only $O(1)$ work. 
Since each expanded call handles only constantly many candidate branches, the exclusive cost of one expanded call is $O(|\domain|^2)$, with the $O(\log |\domain|)$ cost of $\textsc{Extendable}$ absorbed into $O(|\domain|^2)$.

Next we bound the total number of expanded calls. 
Let $C$ be the output circuit and $M$ be the set of non-terminal expanded calls. 
Every $\kappa \in M$ creates a distinct $\lor$-node in $C$: indeed, if two calls had the same abstraction, the later one would be answered by reuse and would not be expanded. 
Moreover, the context of every call in $M$ is extendable. 
Therefore, at the current stage, at least one candidate branch must pass its $\textsc{Extendable}$ test; otherwise the current context would admit no completion, contradicting extendability. 
Hence the $\lor$-node created by $\kappa$ has at least one outgoing edge in $C$.

For each $\kappa \in M$, choose one fixed outgoing edge of its $\lor$-node. 
This defines an injection from $M$ into the edge set of $C$: distinct expanded calls create distinct $\lor$-nodes, and edges with different source gates are distinct edges. 
Therefore, $|M| \le |C|$.

It remains to sum the work. 
By the bound above, each non-terminal expanded call contributes $O(|\domain|^2)$ exclusive time, so the total time spent in all such calls is $O(|M|\,|\domain|^2) = O(|C|\,|\domain|^2)$.
The remaining cost comes from terminal and base cases.
A terminal Stage~II call only returns $\top$ and costs $O(1)$; the Stage~I base case $i=|\domain|+1$ only transfers control to Stage~II and likewise costs $O(1)$.
Such calls occur only along successful branches, together with at most the root call, and are therefore bounded by $O(|C|+1)$.
Hence the total running time is $ O\big((|C|+1)\,|\domain|^2\big)$.
\end{proof}

\section{Evaluation Details}

\subsection{Post-processing}

The basic compilation algorithm may introduce redundant internal logical nodes in the resulting circuit.
For example, for the circuit of $\sentence_E$ in~\cref{sec:experiments}, the Stage I output may contain nested $\wedge$-nodes of the form
$$
\wedge\bigl(\neg E(e_1,e_1) \wedge (\neg E(e_2,e_2) \wedge (\neg E(e_3,e_3) \wedge \dots))\bigr),
$$
which can be flattened without changing the represented Boolean function into
$$
\wedge\bigl(\neg E(e_1,e_1), \neg E(e_2,e_2), \neg E(e_3,e_3), \dots\bigr).
$$

We therefore apply a post-processing step that removes such redundant nodes; see \cref{alg:full-postprocess-simple}. 
Concretely, whenever a non-root $\wedge$-node or $\vee$-node has a unique parent of the same type, we merge the node into its parent. 
This transformation preserves logical equivalence as well as the d-DNNF property.
In the tables below, `ours*' denotes the variant of our compiler with this post-processing step. 
It typically yields smaller circuits, at the cost of some additional running time.

\begin{algorithm}[tb]
    \caption{Remove redundant logical nodes}
    \label{alg:full-postprocess-simple}

    \SetKwInOut{Input}{Input}
    \SetKwInOut{Output}{Output}
    \SetKwProg{Fn}{Function}{ is}{end}

    \Input{A circuit $C$ with root node $r$}
    \Output{An equivalent simplified circuit}

    \BlankLine
    Initialize queue $\mathcal{Q}\gets[r]$ and visited set $\mathcal{V}\gets \{r\}$\;
    \While{$\mathcal{Q}$ is not empty}{
        $u \gets$ pop from $\mathcal{Q}$\;
        push every child of $u$ that is not in $\mathcal{V}$ into $\mathcal{Q}$, and add it to $\mathcal{V}$\;
        \If{$u$ is a logical node with only one parent $p$, and $p$ is of the same type as $u$}{
            add all children of $u$ into children of $p$\;
            delete $u$ and the links involving $u$\;
        }
    }

        \Return $C$\;
\end{algorithm}

\subsection{Correctness validation}

To validate the correctness of our compiler, we performed both model enumeration and counting cross-checks against independent baselines.
For each input pair $(\sentence, \domain)$, we constructed (i) the d-DNNF circuit produced by our compiler, and (ii) a propositional CNF encoding of the grounding of $(\sentence, \domain)$.

For model enumeration, we compared the compiled circuit and the grounding in both directions. First, we enumerated models of the compiled d-DNNF circuit and checked that each enumerated assignment satisfies the grounding CNF. Second, on the SAT side, we used MiniSAT as a SAT oracle on the grounding CNF, repeatedly calling $\mathrm{solve()}$.
After each satisfying assignment $\structure$ was found, we added a blocking clause excluding $\structure$, and continued enumeration. 
In this way, we enumerated satisfying assignments of the grounding CNF and checked that each of them satisfies the compiled d-DNNF circuit. 
Since the number of models can become extremely large for larger domain sizes, exhaustive enumeration is often infeasible; therefore, on each side, we capped enumeration at $100{,}000$ models.

In addition to enumeration, we also validated correctness by comparing model counts obtained from three methods: (i) counting on the d-DNNF compiled by our method, (ii) counting on a d-DNNF compiled from the grounding by the compiler c2d~\cite{c2d}, and (iii) counting using the $\mathrm{WFOMC}$ algorithm for \fotwo{} sentences~\cite{fastWFOMC}.
All three methods returned exactly the same counts on all tested instances.

\subsection{Full evaluation results}

For randomized methods (c2d, bella-cd, and bella-ph), we repeated each experiment until $10$ successful runs had been collected, and report the average circuit size and runtime over those $10$ runs.
If a run timed out, we discarded it and continued until $10$ non-timeout runs were obtained.
For deterministic and stable methods (d4, bella-ka, and ours), we ran each experiment $3$ times and report the average over these runs.
All experiments were executed single-threadedly with a timeout of $1800$ seconds per run. 
For each instance, if the first three runs all timed out, the method was deemed unable to finish within the time limit and we terminated further attempts.

For the representative examples $\sentence_{RB}$, $\sentence_E$, $\sentence_{RBE}$, $\sentence_P$ and $\sentence_{D}$, we evaluated domain sizes from $n=4$ to $n=20$; the full results are reported in \cref{tab:res1-full,tab:res2-full,tab:res3-full,tab:res4-full,tab:res5-full}, respectively. 
For the \texttt{ui-bj} family and its variants, we evaluated domain sizes from $n=4$ to $n=10$; the full results are reported in \cref{tab:u2-b2,tab:u4-b2,tab:u6-b2,tab:u4-b1,tab:u4-b3,tab:u4-b2-ad,tab:u4-b2-so,tab:u4-b2-cc,tab:u4-b2-fd,tab:u4-b2-ao,tab:u4-b2-fo}.
In each table cell, the upper entry reports circuit size, the lower entry reports runtime, and `-' indicates that the first three runs all timed out.
The best results are shown in bold and the second-best results are underlined.

To isolate the effect of structural constraints on the binary predicates, we organize the \texttt{u4-b2} benchmark family into a common core and a set of additional structural constraints. 
The common core is:
\begin{align*}
\texttt{core} = &\forall x: \text{ExactlyOne}\big[ C_1(x), C_2(x), C_3(x), C_4(x)\big] \ \wedge \\
&\forall x\forall y: \big(\bigvee_{k\in[2]} E_k(x,y)\big) \rightarrow \bigwedge_{\ell\in[4]} \big(\neg C_\ell(x) \vee \neg C_\ell(y)\big) \ \wedge \\
&\bigwedge_{k\in[2]} \forall x\exists y: E_k(x,y).
\end{align*}
This core imposes no constraint on the relation between $E_k(x,y)$ and $E_k(y,x)$, for a fixed $k$, and it also imposes no interaction constraint between $E_1$ and $E_2$.
We then derive several variants by adding constraints along two orthogonal dimensions:
\begin{description}
\item[Within-predicate orientation pattern.] Each binary predicate $E_k$ may be 
\begin{itemize}
  \item symmetric: $E_k(x,y)\rightarrow E_k(y,x)$;
  \item asymmetric: $E_k(x,y)\rightarrow \neg E_k(y,x)$; or
  \item free, i.e., no symmetry or asymmetry constraint is imposed.
\end{itemize}
\item[Cross-predicate interaction pattern.] The two binary predicates $E_1$ and $E_2$ may be:
\begin{itemize}
  \item same-direction disjointness: $\neg E_1(x,y) \vee \neg E_2(x,y)$;
  \item bidirectional disjointness: 
  $$\big(E_1(x,y) \rightarrow (\neg E_2(x,y) \wedge \neg E_2(y,x))\big) \wedge \big(E_2(x,y) \rightarrow (\neg E_1(x,y) \wedge \neg E_1(y,x))\big);$$
  \item converse coupling: $\big(E_1(x,y) \rightarrow E_2(y,x)\big) \wedge \big(E_2(x,y) \rightarrow E_1(y,x)\big)$; or
  \item overlap allowed, i.e., no disjointness constraint is imposed.
\end{itemize}
\end{description}
Using these two dimensions, we define the following benchmark variants:
\begin{itemize}
  \item the symmetric-disjoint variant: both $E_1$ and $E_2$ are symmetric, and they cannot co-occur on the same ordered pair, which is the default setting in \texttt{ui-bj}.
  \begin{align*}
  \texttt{u4-b2} = 
  & \texttt{core} \wedge \forall x\forall y: \big(E_1(x,y) \rightarrow E_1(y,x)\big) \wedge \big(E_2(x,y) \rightarrow E_2(y,x)\big) \ \wedge \\
  &\forall x\forall y: \neg E_1(x,y) \vee \neg E_2(x,y).
  \end{align*}
  \item the asymmetric-disjoint variant: both predicates are asymmetric, and they cannot co-occur on the same ordered pair.
  \begin{align*}
  \texttt{u4-b2-ad} =
  & \texttt{core} \wedge \forall x\forall y: \big(E_1(x,y) \rightarrow \neg E_1(y,x)\big) \wedge \big(E_2(x,y) \rightarrow \neg E_2(y,x)\big) \ \wedge \\
  &\forall x\forall y: \neg E_1(x,y) \vee \neg E_2(x,y).
  \end{align*}
  \item the symmetric-overlap variant: both predicates are symmetric, but overlap between them is allowed.
  \begin{align*}
  \texttt{u4-b2-so} = 
  & \texttt{core} \wedge \forall x\forall y: \big(E_1(x,y) \rightarrow E_1(y,x)\big) \wedge \big(E_2(x,y) \rightarrow E_2(y,x)\big).
  \end{align*}
  \item the converse-coupled variant: an $E_1$-edge in one direction forces an $E_2$-edge in the opposite direction, and symmetrically for $E_2$.
  \begin{align*}
  \texttt{u4-b2-cc} = 
  & \texttt{core} \wedge \forall x\forall y: E_1(x,y) \rightarrow \big(E_2(y,x) \wedge \neg E_1(y,x)\big) \wedge \\ 
  &\forall x\forall y: E_2(x,y) \rightarrow \big(E_1(y,x) \wedge \neg E_2(y,x)\big).
  \end{align*} 
  \item the free-bidirectionally-disjoint variant: neither $E_1$ nor $E_2$ is required to be symmetric or asymmetric, but the two predicates are mutually exclusive across both directions of each unordered pair.
  \begin{align*}
  \texttt{u4-b2-fd} = 
  & \texttt{core} \wedge \forall x\forall y: E_1(x,y) \rightarrow \big(\neg E_2(x,y) \wedge \neg E_2(y,x)\big) \wedge \\
  &  \forall x\forall y: E_2(x,y) \rightarrow \big(\neg E_1(x,y) \wedge \neg E_1(y,x)\big).
  \end{align*}
  \item the asymmetric-overlap variant: both predicates are asymmetric, while overlap between $E_1$ and $E_2$ is allowed.
  \begin{align*}
  \texttt{u4-b2-ao} = 
  & \texttt{core} \wedge \forall x\forall y: \big(E_1(x,y) \rightarrow \neg E_1(y,x)\big) \wedge \big(E_2(x,y) \rightarrow \neg E_2(y,x)\big).
  \end{align*}
  \item the free-overlap variant, in which no additional structural constraint is imposed beyond the common core.
  \begin{align*}
  \texttt{u4-b2-fo} = 
  & \texttt{core}.
  \end{align*}
\end{itemize}

\begin{table}[t]
\centering
\small
\caption{The circuit sizes and compilation times of $\sentence_{RB}$. Interestingly, Bella failed on the $n=15$ instance (timeout) but succeeded on both smaller and larger domain sizes under this sentence. We therefore treat this as an anomalous behavior.}
\begin{tabular}{c|ccccccc}
\toprule

n & ours & ours* & bella-cd & bella-ka & bella-ph & d4 & c2d \\
\midrule
% \raisebox{1.30ex}{\strut 3} & \shortstack{\strut \underline{94}\\[-0.85ex]\strut 0.00} & \shortstack{\strut \textbf{88}\\[-0.85ex]\strut 0.00} & \shortstack{\strut 121\\[-0.85ex]\strut 0.00} & \shortstack{\strut 121\\[-0.85ex]\strut 0.00} & \shortstack{\strut 121\\[-0.85ex]\strut 0.00} & \shortstack{\strut 104\\[-0.85ex]\strut 0.00} & \shortstack{\strut 122\\[-0.85ex]\strut 0.09} \\
% \midrule
\raisebox{1.30ex}{\strut 4} & \shortstack{\strut \underline{214}\\[-0.85ex]\strut 0.00} & \shortstack{\strut \textbf{192}\\[-0.85ex]\strut 0.00} & \shortstack{\strut 298\\[-0.85ex]\strut 0.00} & \shortstack{\strut 298\\[-0.85ex]\strut 0.00} & \shortstack{\strut 298\\[-0.85ex]\strut 0.00} & \shortstack{\strut 260\\[-0.85ex]\strut 0.00} & \shortstack{\strut 300\\[-0.85ex]\strut 0.16} \\
\midrule
\raisebox{1.30ex}{\strut 5} & \shortstack{\strut \underline{468}\\[-0.85ex]\strut 0.00} & \shortstack{\strut \textbf{398}\\[-0.85ex]\strut 0.00} & \shortstack{\strut 689\\[-0.85ex]\strut 0.01} & \shortstack{\strut 689\\[-0.85ex]\strut 0.01} & \shortstack{\strut 689\\[-0.85ex]\strut 0.00} & \shortstack{\strut 616\\[-0.85ex]\strut 0.00} & \shortstack{\strut 680\\[-0.85ex]\strut 0.27} \\
\midrule
\raisebox{1.30ex}{\strut 6} & \shortstack{\strut \underline{1016}\\[-0.85ex]\strut 0.00} & \shortstack{\strut \textbf{818}\\[-0.85ex]\strut 0.00} & \shortstack{\strut 1550\\[-0.85ex]\strut 0.03} & \shortstack{\strut 1550\\[-0.85ex]\strut 0.03} & \shortstack{\strut 1550\\[-0.85ex]\strut 0.01} & \shortstack{\strut 1410\\[-0.85ex]\strut 0.00} & \shortstack{\strut 1498\\[-0.85ex]\strut 0.36} \\
\bottomrule
\end{tabular}
\label{tab:res1-full}
\end{table}

\begin{table}[t]
\centering
\small
\caption{The circuit sizes and compilation times of $\sentence_{E}$}
\begin{tabular}{c|ccccccc}
\toprule

n & ours & ours* & bella-cd & bella-ka & bella-ph & d4 & c2d \\
\midrule
% \raisebox{1.30ex}{\strut 3} & \shortstack{\strut 34\\[-0.85ex]\strut 0.00} & \shortstack{\strut \underline{28}\\[-0.85ex]\strut 0.00} & \shortstack{\strut 35\\[-0.85ex]\strut 0.00} & \shortstack{\strut 35\\[-0.85ex]\strut 0.00} & \shortstack{\strut 35\\[-0.85ex]\strut 0.00} & \shortstack{\strut \textbf{24}\\[-0.85ex]\strut 0.00} & \shortstack{\strut 33\\[-0.85ex]\strut 0.04} \\
% \midrule
\raisebox{1.30ex}{\strut 4} & \shortstack{\strut 116\\[-0.85ex]\strut 0.00} & \shortstack{\strut \textbf{107}\\[-0.85ex]\strut 0.00} & \shortstack{\strut 162\\[-0.85ex]\strut 0.00} & \shortstack{\strut 162\\[-0.85ex]\strut 0.00} & \shortstack{\strut 147\\[-0.85ex]\strut 0.00} & \shortstack{\strut \underline{114}\\[-0.85ex]\strut 0.00} & \shortstack{\strut 136\\[-0.85ex]\strut 0.05} \\
\midrule
\raisebox{1.30ex}{\strut 5} & \shortstack{\strut \underline{336}\\[-0.85ex]\strut 0.00} & \shortstack{\strut \textbf{319}\\[-0.85ex]\strut 0.00} & \shortstack{\strut 439\\[-0.85ex]\strut 0.00} & \shortstack{\strut 439\\[-0.85ex]\strut 0.00} & \shortstack{\strut 459\\[-0.85ex]\strut 0.00} & \shortstack{\strut 371\\[-0.85ex]\strut 0.00} & \shortstack{\strut 460\\[-0.85ex]\strut 0.07} \\
\midrule
\raisebox{1.30ex}{\strut 6} & \shortstack{\strut \underline{890}\\[-0.85ex]\strut 0.00} & \shortstack{\strut \textbf{854}\\[-0.85ex]\strut 0.00} & \shortstack{\strut 1360\\[-0.85ex]\strut 0.01} & \shortstack{\strut 1347\\[-0.85ex]\strut 0.01} & \shortstack{\strut 1362\\[-0.85ex]\strut 0.01} & \shortstack{\strut 1009\\[-0.85ex]\strut 0.00} & \shortstack{\strut 1480\\[-0.85ex]\strut 0.09} \\
\bottomrule
\end{tabular}
\label{tab:res2-full}
\end{table}

\begin{table}[t]
\centering
\small
\caption{The circuit sizes and compilation times of $\sentence_{RBE}$}
\begin{tabular}{c|ccccccc}
\toprule

n & ours & ours* & bella-cd & bella-ka & bella-ph & d4 & c2d \\
\midrule
% \raisebox{1.30ex}{\strut 3} & \shortstack{\strut \underline{76}\\[-0.85ex]\strut 0.00} & \shortstack{\strut \textbf{66}\\[-0.85ex]\strut 0.00} & \shortstack{\strut 84\\[-0.85ex]\strut 0.00} & \shortstack{\strut 84\\[-0.85ex]\strut 0.00} & \shortstack{\strut 84\\[-0.85ex]\strut 0.00} & \shortstack{\strut \textbf{66}\\[-0.85ex]\strut 0.00} & \shortstack{\strut 76\\[-0.85ex]\strut 0.07} \\
% \midrule
\raisebox{1.30ex}{\strut 4} & \shortstack{\strut \underline{260}\\[-0.85ex]\strut 0.00} & \shortstack{\strut \textbf{232}\\[-0.85ex]\strut 0.00} & \shortstack{\strut 634\\[-0.85ex]\strut 0.01} & \shortstack{\strut 634\\[-0.85ex]\strut 0.01} & \shortstack{\strut 583\\[-0.85ex]\strut 0.00} & \shortstack{\strut 343\\[-0.85ex]\strut 0.00} & \shortstack{\strut 414\\[-0.85ex]\strut 0.10} \\
\midrule
\raisebox{1.30ex}{\strut 5} & \shortstack{\strut \underline{1026}\\[-0.85ex]\strut 0.01} & \shortstack{\strut \textbf{901}\\[-0.85ex]\strut 0.01} & \shortstack{\strut 2032\\[-0.85ex]\strut 0.02} & \shortstack{\strut 2199\\[-0.85ex]\strut 0.02} & \shortstack{\strut 2085\\[-0.85ex]\strut 0.01} & \shortstack{\strut 1596\\[-0.85ex]\strut 0.00} & \shortstack{\strut 2180\\[-0.85ex]\strut 0.16} \\
\midrule
\raisebox{1.30ex}{\strut 6} & \shortstack{\strut \underline{4518}\\[-0.85ex]\strut 0.02} & \shortstack{\strut \textbf{3895}\\[-0.85ex]\strut 0.02} & \shortstack{\strut 7095\\[-0.85ex]\strut 0.03} & \shortstack{\strut 8454\\[-0.85ex]\strut 0.05} & \shortstack{\strut 7428\\[-0.85ex]\strut 0.02} & \shortstack{\strut 7217\\[-0.85ex]\strut 0.01} & \shortstack{\strut 7297\\[-0.85ex]\strut 0.19} \\
\bottomrule
\end{tabular}
\label{tab:res3-full}
\end{table}

\begin{table}[t]
\centering
\small
\caption{The circuit sizes and compilation times of $\sentence_{P}$}
\begin{tabular}{c|ccccccc}
\toprule

n & ours & ours* & bella-cd & bella-ka & bella-ph & d4 & c2d \\
\midrule
% \raisebox{1.30ex}{\strut 3} & \shortstack{\strut 160\\[-0.85ex]\strut 0.00} & \shortstack{\strut 157\\[-0.85ex]\strut 0.00} & \shortstack{\strut 146\\[-0.85ex]\strut 0.00} & \shortstack{\strut 146\\[-0.85ex]\strut 0.00} & \shortstack{\strut 147\\[-0.85ex]\strut 0.00} & \shortstack{\strut \textbf{86}\\[-0.85ex]\strut 0.00} & \shortstack{\strut \underline{118}\\[-0.85ex]\strut 0.03} \\
% \midrule
\raisebox{1.30ex}{\strut 4} & \shortstack{\strut 1162\\[-0.85ex]\strut 0.01} & \shortstack{\strut 1147\\[-0.85ex]\strut 0.01} & \shortstack{\strut 766\\[-0.85ex]\strut 0.01} & \shortstack{\strut 774\\[-0.85ex]\strut 0.01} & \shortstack{\strut 773\\[-0.85ex]\strut 0.00} & \shortstack{\strut \textbf{590}\\[-0.85ex]\strut 0.00} & \shortstack{\strut \underline{762}\\[-0.85ex]\strut 0.04} \\
\midrule
\raisebox{1.30ex}{\strut 5} & \shortstack{\strut 6586\\[-0.85ex]\strut 0.04} & \shortstack{\strut 6539\\[-0.85ex]\strut 0.04} & \shortstack{\strut \underline{3643}\\[-0.85ex]\strut 0.01} & \shortstack{\strut \underline{3643}\\[-0.85ex]\strut 0.01} & \shortstack{\strut 3689\\[-0.85ex]\strut 0.01} & \shortstack{\strut \textbf{2565}\\[-0.85ex]\strut 0.00} & \shortstack{\strut 4844\\[-0.85ex]\strut 0.05} \\
\midrule
\raisebox{1.30ex}{\strut 6} & \shortstack{\strut 33620\\[-0.85ex]\strut 0.09} & \shortstack{\strut 33473\\[-0.85ex]\strut 0.09} & \shortstack{\strut \underline{17032}\\[-0.85ex]\strut 0.04} & \shortstack{\strut 17043\\[-0.85ex]\strut 0.07} & \shortstack{\strut 17977\\[-0.85ex]\strut 0.04} & \shortstack{\strut \textbf{11112}\\[-0.85ex]\strut 0.01} & \shortstack{\strut 37398\\[-0.85ex]\strut 0.08} \\
\bottomrule
\end{tabular}
\label{tab:res4-full}
\end{table}

\begin{table}[t]
\centering
\small
\caption{The circuit sizes and compilation times of $\sentence_{D}$}
\begin{tabular}{c|ccccccc}
\toprule

n & ours & ours* & bella-cd & bella-ka & bella-ph & d4 & c2d \\
\midrule
\raisebox{1.30ex}{\strut 4} & \shortstack{\strut \underline{332}\\[-0.85ex]\strut 0.00} & \shortstack{\strut \textbf{299}\\[-0.85ex]\strut 0.00} & \shortstack{\strut 432\\[-0.85ex]\strut 0.01} & \shortstack{\strut 432\\[-0.85ex]\strut 0.01} & \shortstack{\strut 465\\[-0.85ex]\strut 0.00} & \shortstack{\strut 366\\[-0.85ex]\strut 0.00} & \shortstack{\strut 389\\[-0.85ex]\strut 0.14} \\
\midrule
\raisebox{1.30ex}{\strut 5} & \shortstack{\strut \underline{1146}\\[-0.85ex]\strut 0.01} & \shortstack{\strut \textbf{1006}\\[-0.85ex]\strut 0.01} & \shortstack{\strut 1572\\[-0.85ex]\strut 0.02} & \shortstack{\strut 1558\\[-0.85ex]\strut 0.02} & \shortstack{\strut 1610\\[-0.85ex]\strut 0.01} & \shortstack{\strut 1183\\[-0.85ex]\strut 0.00} & \shortstack{\strut 1753\\[-0.85ex]\strut 0.14} \\
\midrule
\raisebox{1.30ex}{\strut 6} & \shortstack{\strut \underline{4154}\\[-0.85ex]\strut 0.02} & \shortstack{\strut \textbf{3576}\\[-0.85ex]\strut 0.02} & \shortstack{\strut 5821\\[-0.85ex]\strut 0.04} & \shortstack{\strut 5851\\[-0.85ex]\strut 0.04} & \shortstack{\strut 5948\\[-0.85ex]\strut 0.03} & \shortstack{\strut 4362\\[-0.85ex]\strut 0.01} & \shortstack{\strut 6712\\[-0.85ex]\strut 0.30} \\
\midrule
\raisebox{1.30ex}{\strut 7} & \shortstack{\strut \underline{15174}\\[-0.85ex]\strut 0.04} & \shortstack{\strut \textbf{12889}\\[-0.85ex]\strut 0.03} & \shortstack{\strut 21916\\[-0.85ex]\strut 0.10} & \shortstack{\strut 21916\\[-0.85ex]\strut 0.10} & \shortstack{\strut 22037\\[-0.85ex]\strut 0.08} & \shortstack{\strut 17546\\[-0.85ex]\strut 0.03} & \shortstack{\strut 60084\\[-0.85ex]\strut 0.74} \\
\bottomrule
\end{tabular}
\label{tab:res5-full}
\end{table}

\begin{table}[t]
\centering
\small
\caption{The circuit sizes and compilation times of \texttt{u4-b2}}
\begin{tabular}{c|ccccccc}
\toprule

n & ours & ours* & bella-cd & bella-ka & bella-ph & d4 & c2d \\
\midrule
\raisebox{1.30ex}{\strut 4} & \shortstack{\strut \underline{951}\\[-0.85ex]\strut 0.02} & \shortstack{\strut \textbf{923}\\[-0.85ex]\strut 0.02} & \shortstack{\strut 6833\\[-0.85ex]\strut 0.07} & \shortstack{\strut 6836\\[-0.85ex]\strut 0.10} & \shortstack{\strut 6856\\[-0.85ex]\strut 0.06} & \shortstack{\strut 7359\\[-0.85ex]\strut 0.03} & \shortstack{\strut 5380\\[-0.85ex]\strut 0.39} \\
\midrule
\raisebox{1.30ex}{\strut 5} & \shortstack{\strut \underline{10699}\\[-0.85ex]\strut 0.07} & \shortstack{\strut \textbf{10380}\\[-0.85ex]\strut 0.07} & \shortstack{\strut 109832\\[-0.85ex]\strut 0.87} & \shortstack{\strut 111596\\[-0.85ex]\strut 1.24} & \shortstack{\strut 109056\\[-0.85ex]\strut 0.82} & \shortstack{\strut 213827\\[-0.85ex]\strut 0.50} & \shortstack{\strut 54955\\[-0.85ex]\strut 1.48} \\
\midrule
\raisebox{1.30ex}{\strut 6} & \shortstack{\strut \underline{200829}\\[-0.85ex]\strut 0.66} & \shortstack{\strut \textbf{194328}\\[-0.85ex]\strut 0.66} & \shortstack{\strut 6581517\\[-0.85ex]\strut 95.23} & \shortstack{\strut 7082727\\[-0.85ex]\strut 104.04} & \shortstack{\strut 2889151\\[-0.85ex]\strut 17.39} & \shortstack{\strut 3279954\\[-0.85ex]\strut 6.99} & \shortstack{\strut 873668\\[-0.85ex]\strut 9.59} \\
\midrule
\raisebox{1.30ex}{\strut 7} & \shortstack{\strut \underline{3703705}\\[-0.85ex]\strut \textbf{7.46}} & \shortstack{\strut \textbf{3583660}\\[-0.85ex]\strut \underline{10.01}} & \raisebox{1.30ex}{\strut -} & \raisebox{1.30ex}{\strut -} & \raisebox{1.30ex}{\strut -} & \shortstack{\strut 176233803\\[-0.85ex]\strut 899.99} & \shortstack{\strut 15562022\\[-0.85ex]\strut 145.34} \\
\bottomrule
\end{tabular}
\label{tab:u4-b2}
\end{table}

\begin{table}[t]
\centering
\small
\caption{The circuit sizes and compilation times of \texttt{u2-b2}}
\begin{tabular}{c|ccccccc}
\toprule

n & ours & ours* & bella-cd & bella-ka & bella-ph & d4 & c2d \\
\midrule
\raisebox{1.30ex}{\strut 4} & \shortstack{\strut \underline{214}\\[-0.85ex]\strut 0.00} & \shortstack{\strut \textbf{181}\\[-0.85ex]\strut 0.00} & \shortstack{\strut 388\\[-0.85ex]\strut 0.01} & \shortstack{\strut 388\\[-0.85ex]\strut 0.01} & \shortstack{\strut 356\\[-0.85ex]\strut 0.00} & \shortstack{\strut 256\\[-0.85ex]\strut 0.00} & \shortstack{\strut 396\\[-0.85ex]\strut 0.31} \\
\midrule
\raisebox{1.30ex}{\strut 5} & \shortstack{\strut \underline{866}\\[-0.85ex]\strut 0.00} & \shortstack{\strut \textbf{729}\\[-0.85ex]\strut 0.01} & \shortstack{\strut 2489\\[-0.85ex]\strut 0.02} & \shortstack{\strut 2014\\[-0.85ex]\strut 0.02} & \shortstack{\strut 2439\\[-0.85ex]\strut 0.01} & \shortstack{\strut 1434\\[-0.85ex]\strut 0.00} & \shortstack{\strut 1778\\[-0.85ex]\strut 0.44} \\
\midrule
\raisebox{1.30ex}{\strut 6} & \shortstack{\strut \underline{8057}\\[-0.85ex]\strut 0.04} & \shortstack{\strut \textbf{6886}\\[-0.85ex]\strut \underline{0.02}} & \shortstack{\strut 23195\\[-0.85ex]\strut 0.09} & \shortstack{\strut 24417\\[-0.85ex]\strut 0.33} & \shortstack{\strut 32530\\[-0.85ex]\strut 0.07} & \shortstack{\strut 19524\\[-0.85ex]\strut \textbf{0.01}} & \shortstack{\strut 13053\\[-0.85ex]\strut 0.62} \\
\midrule
\raisebox{1.30ex}{\strut 7} & \shortstack{\strut \underline{84605}\\[-0.85ex]\strut \textbf{0.17}} & \shortstack{\strut \textbf{73827}\\[-0.85ex]\strut \underline{0.18}} & \shortstack{\strut 281697\\[-0.85ex]\strut 0.75} & \shortstack{\strut 356807\\[-0.85ex]\strut 4.25} & \shortstack{\strut 352267\\[-0.85ex]\strut 0.82} & \shortstack{\strut 433123\\[-0.85ex]\strut 0.46} & \shortstack{\strut 115994\\[-0.85ex]\strut 0.93} \\
\bottomrule
\end{tabular}
\label{tab:u2-b2}
\end{table}

\begin{table}[t]
\centering
\small
\caption{The circuit sizes and compilation times of \texttt{u6-b2}}
\begin{tabular}{c|ccccccc}
\toprule

n & ours & ours* & bella-cd & bella-ka & bella-ph & d4 & c2d \\
\midrule
\raisebox{1.30ex}{\strut 4} & \shortstack{\strut \underline{2569}\\[-0.85ex]\strut 0.04} & \shortstack{\strut \textbf{2539}\\[-0.85ex]\strut 0.04} & \shortstack{\strut 22053\\[-0.85ex]\strut 0.07} & \shortstack{\strut 21814\\[-0.85ex]\strut 0.19} & \shortstack{\strut 24935\\[-0.85ex]\strut 0.07} & \shortstack{\strut 25076\\[-0.85ex]\strut 0.02} & \shortstack{\strut 25633\\[-0.85ex]\strut 0.41} \\
\midrule
\raisebox{1.30ex}{\strut 5} & \shortstack{\strut \underline{21500}\\[-0.85ex]\strut \underline{0.11}} & \shortstack{\strut \textbf{21197}\\[-0.85ex]\strut \textbf{0.09}} & \shortstack{\strut 1038502\\[-0.85ex]\strut 2.78} & \shortstack{\strut 998853\\[-0.85ex]\strut 6.58} & \shortstack{\strut 1177608\\[-0.85ex]\strut 2.84} & \shortstack{\strut 3149975\\[-0.85ex]\strut 6.93} & \shortstack{\strut 865271\\[-0.85ex]\strut 2.90} \\
\midrule
\raisebox{1.30ex}{\strut 6} & \shortstack{\strut \underline{281037}\\[-0.85ex]\strut \underline{0.90}} & \shortstack{\strut \textbf{275106}\\[-0.85ex]\strut \textbf{0.68}} & \shortstack{\strut 41116852\\[-0.85ex]\strut 99.38} & \shortstack{\strut 40093317\\[-0.85ex]\strut 122.35} & \shortstack{\strut 87660453\\[-0.85ex]\strut 204.51} & \shortstack{\strut 133224979\\[-0.85ex]\strut 154.21} & \shortstack{\strut 94448888\\[-0.85ex]\strut 390.33} \\
\midrule
\raisebox{1.30ex}{\strut 7} & \shortstack{\strut \underline{4692680}\\[-0.85ex]\strut \textbf{10.63}} & \shortstack{\strut \textbf{4581060}\\[-0.85ex]\strut \underline{13.25}} & \raisebox{1.30ex}{\strut -} & \raisebox{1.30ex}{\strut -} & \raisebox{1.30ex}{\strut -} & \raisebox{1.30ex}{\strut -} & \raisebox{1.30ex}{\strut -} \\
\bottomrule
\end{tabular}
\label{tab:u6-b2}
\end{table}

\begin{table}[t]
\centering
\small
\caption{The circuit sizes and compilation times of \texttt{u4-b1}}
\begin{tabular}{c|ccccccc}
\toprule

n & ours & ours* & bella-cd & bella-ka & bella-ph & d4 & c2d \\
\midrule
\raisebox{1.30ex}{\strut 4} & \shortstack{\strut \underline{1012}\\[-0.85ex]\strut 0.01} & \shortstack{\strut \textbf{989}\\[-0.85ex]\strut 0.01} & \shortstack{\strut 4428\\[-0.85ex]\strut 0.04} & \shortstack{\strut 4428\\[-0.85ex]\strut 0.04} & \shortstack{\strut 4625\\[-0.85ex]\strut 0.04} & \shortstack{\strut 4355\\[-0.85ex]\strut 0.01} & \shortstack{\strut 3926\\[-0.85ex]\strut 0.21} \\
\midrule
\raisebox{1.30ex}{\strut 5} & \shortstack{\strut \underline{5378}\\[-0.85ex]\strut \textbf{0.02}} & \shortstack{\strut \textbf{5212}\\[-0.85ex]\strut \underline{0.03}} & \shortstack{\strut 33350\\[-0.85ex]\strut 0.20} & \shortstack{\strut 33333\\[-0.85ex]\strut 0.22} & \shortstack{\strut 38791\\[-0.85ex]\strut 0.26} & \shortstack{\strut 29059\\[-0.85ex]\strut 0.08} & \shortstack{\strut 20582\\[-0.85ex]\strut 0.46} \\
\midrule
\raisebox{1.30ex}{\strut 6} & \shortstack{\strut \underline{35530}\\[-0.85ex]\strut \textbf{0.06}} & \shortstack{\strut \textbf{33931}\\[-0.85ex]\strut \underline{0.07}} & \shortstack{\strut 341864\\[-0.85ex]\strut 3.99} & \shortstack{\strut 325541\\[-0.85ex]\strut 3.83} & \shortstack{\strut 248137\\[-0.85ex]\strut 1.39} & \shortstack{\strut 217450\\[-0.85ex]\strut 0.61} & \shortstack{\strut 247820\\[-0.85ex]\strut 3.19} \\
\midrule
\raisebox{1.30ex}{\strut 7} & \shortstack{\strut \underline{269728}\\[-0.85ex]\strut \underline{0.30}} & \shortstack{\strut \textbf{254248}\\[-0.85ex]\strut \textbf{0.29}} & \shortstack{\strut 4087266\\[-0.85ex]\strut 40.62} & \shortstack{\strut 3675504\\[-0.85ex]\strut 37.59} & \shortstack{\strut 3603752\\[-0.85ex]\strut 61.10} & \shortstack{\strut 2731890\\[-0.85ex]\strut 10.23} & \shortstack{\strut 1709102\\[-0.85ex]\strut 28.93} \\
\bottomrule
\end{tabular}
\label{tab:u4-b1}
\end{table}

\begin{table}[t]
\centering
\small
\caption{The circuit sizes and compilation times of \texttt{u4-b3}}
\begin{tabular}{c|ccccccc}
\toprule

n & ours & ours* & bella-cd & bella-ka & bella-ph & d4 & c2d \\
\midrule
\raisebox{1.30ex}{\strut 4} & \shortstack{\strut \underline{375}\\[-0.85ex]\strut 0.02} & \shortstack{\strut \textbf{346}\\[-0.85ex]\strut 0.02} & \shortstack{\strut 2638\\[-0.85ex]\strut 0.05} & \shortstack{\strut 2638\\[-0.85ex]\strut 0.05} & \shortstack{\strut 2000\\[-0.85ex]\strut 0.04} & \shortstack{\strut 3811\\[-0.85ex]\strut 0.02} & \shortstack{\strut 2107\\[-0.85ex]\strut 0.73} \\
\midrule
\raisebox{1.30ex}{\strut 5} & \shortstack{\strut \underline{4512}\\[-0.85ex]\strut 0.08} & \shortstack{\strut \textbf{4150}\\[-0.85ex]\strut 0.08} & \shortstack{\strut 110055\\[-0.85ex]\strut 0.87} & \shortstack{\strut 113330\\[-0.85ex]\strut 1.00} & \shortstack{\strut 110200\\[-0.85ex]\strut 0.73} & \shortstack{\strut 172855\\[-0.85ex]\strut 0.38} & \shortstack{\strut 55806\\[-0.85ex]\strut 1.91} \\
\midrule
\raisebox{1.30ex}{\strut 6} & \shortstack{\strut \underline{496131}\\[-0.85ex]\strut \underline{3.16}} & \shortstack{\strut \textbf{477786}\\[-0.85ex]\strut \textbf{3.06}} & \shortstack{\strut 25924673\\[-0.85ex]\strut 382.00} & \shortstack{\strut 24733494\\[-0.85ex]\strut 421.84} & \shortstack{\strut 15193895\\[-0.85ex]\strut 117.43} & \shortstack{\strut 34940307\\[-0.85ex]\strut 118.79} & \shortstack{\strut 2034931\\[-0.85ex]\strut 29.56} \\
\midrule
\raisebox{1.30ex}{\strut 7} & \shortstack{\strut \underline{26671824}\\[-0.85ex]\strut \textbf{102.27}} & \shortstack{\strut \textbf{25962350}\\[-0.85ex]\strut \underline{116.02}} & \raisebox{1.30ex}{\strut -} & \raisebox{1.30ex}{\strut -} & \raisebox{1.30ex}{\strut -} & \raisebox{1.30ex}{\strut -} & \shortstack{\strut 80220334\\[-0.85ex]\strut 1182.81} \\

\bottomrule
\end{tabular}
\label{tab:u4-b3}
\end{table}

\begin{table}[t]
\centering
\small
\caption{The circuit sizes and compilation times of \texttt{u4-b2-ad}}
\begin{tabular}{c|ccccccc}
\toprule

n & ours & ours* & bella-cd & bella-ka & bella-ph & d4 & c2d \\
\midrule
\raisebox{1.30ex}{\strut 6} & \shortstack{\strut \underline{158498}\\[-0.85ex]\strut \textbf{0.90}} & \shortstack{\strut \textbf{154069}\\[-0.85ex]\strut \underline{0.91}} & \shortstack{\strut 8353802\\[-0.85ex]\strut 125.91} & \shortstack{\strut 7234839\\[-0.85ex]\strut 169.81} & \shortstack{\strut 4174572\\[-0.85ex]\strut 34.73} & \shortstack{\strut 8120492\\[-0.85ex]\strut 50.65} & \shortstack{\strut 759082\\[-0.85ex]\strut 11.54} \\
\midrule
\raisebox{1.30ex}{\strut 7} & \shortstack{\strut \underline{5029113}\\[-0.85ex]\strut \textbf{12.20}} & \shortstack{\strut \textbf{4917516}\\[-0.85ex]\strut \underline{14.31}} & \raisebox{1.30ex}{\strut -} & \raisebox{1.30ex}{\strut -} & \raisebox{1.30ex}{\strut -} & \raisebox{1.30ex}{\strut -} & \shortstack{\strut 15667474\\[-0.85ex]\strut 263.72} \\
\midrule
\raisebox{1.30ex}{\strut 8} & \shortstack{\strut \underline{100049630}\\[-0.85ex]\strut \textbf{166.77}} & \shortstack{\strut \textbf{98203765}\\[-0.85ex]\strut \underline{413.40}} & \raisebox{1.30ex}{\strut -} & \raisebox{1.30ex}{\strut -} & \raisebox{1.30ex}{\strut -} & \raisebox{1.30ex}{\strut -} & \raisebox{1.30ex}{\strut -} \\

\bottomrule
\end{tabular}
\label{tab:u4-b2-ad}
\end{table}

\begin{table}[t]
\centering
\small
\caption{The circuit sizes and compilation times of \texttt{u4-b2-cc}}
\begin{tabular}{c|ccccccc}
\toprule

n & ours & ours* & bella-cd & bella-ka & bella-ph & d4 & c2d \\
\midrule
\raisebox{1.30ex}{\strut 4} & \shortstack{\strut \underline{1021}\\[-0.85ex]\strut 0.02} & \shortstack{\strut \textbf{997}\\[-0.85ex]\strut 0.02} & \shortstack{\strut 8544\\[-0.85ex]\strut 0.11} & \shortstack{\strut 8309\\[-0.85ex]\strut 0.31} & \shortstack{\strut 7944\\[-0.85ex]\strut 0.09} & \shortstack{\strut 6932\\[-0.85ex]\strut 0.03} & \shortstack{\strut 5496\\[-0.85ex]\strut 0.42} \\
\midrule
\raisebox{1.30ex}{\strut 5} & \shortstack{\strut \underline{10644}\\[-0.85ex]\strut \underline{0.05}} & \shortstack{\strut \textbf{10307}\\[-0.85ex]\strut \textbf{0.04}} & \shortstack{\strut 114244\\[-0.85ex]\strut 0.94} & \shortstack{\strut 113826\\[-0.85ex]\strut 1.46} & \shortstack{\strut 105048\\[-0.85ex]\strut 0.82} & \shortstack{\strut 106289\\[-0.85ex]\strut 0.35} & \shortstack{\strut 58761\\[-0.85ex]\strut 1.55} \\
\midrule
\raisebox{1.30ex}{\strut 6} & \shortstack{\strut \underline{190503}\\[-0.85ex]\strut \textbf{0.60}} & \shortstack{\strut \textbf{184083}\\[-0.85ex]\strut \underline{0.68}} & \shortstack{\strut 6513604\\[-0.85ex]\strut 95.32} & \shortstack{\strut 6513024\\[-0.85ex]\strut 97.32} & \shortstack{\strut 1751020\\[-0.85ex]\strut 10.13} & \shortstack{\strut 2435471\\[-0.85ex]\strut 5.75} & \shortstack{\strut 921947\\[-0.85ex]\strut 9.95} \\
\midrule
\raisebox{1.30ex}{\strut 7} & \shortstack{\strut \underline{3521784}\\[-0.85ex]\strut \textbf{6.40}} & \shortstack{\strut \textbf{3406724}\\[-0.85ex]\strut \underline{8.26}} & \raisebox{1.30ex}{\strut -} & \raisebox{1.30ex}{\strut -} & \raisebox{1.30ex}{\strut -} & \shortstack{\strut 215909430\\[-0.85ex]\strut 978.79} & \shortstack{\strut 16610579\\[-0.85ex]\strut 154.45} \\
\bottomrule
\end{tabular}
\label{tab:u4-b2-cc}
\end{table}

\begin{table}[t]
\centering
\small
\caption{The circuit sizes and compilation times of \texttt{u4-b2-fd}}
\begin{tabular}{c|ccccccc}
\toprule

n & ours & ours* & bella-cd & bella-ka & bella-ph & d4 & c2d \\
\midrule
\raisebox{1.30ex}{\strut 4} & \shortstack{\strut \underline{1732}\\[-0.85ex]\strut 0.04} & \shortstack{\strut \textbf{1688}\\[-0.85ex]\strut 0.04} & \shortstack{\strut 17721\\[-0.85ex]\strut 0.17} & \shortstack{\strut 17695\\[-0.85ex]\strut 0.43} & \shortstack{\strut 21422\\[-0.85ex]\strut 0.14} & \shortstack{\strut 18277\\[-0.85ex]\strut 0.05} & \shortstack{\strut 8959\\[-0.85ex]\strut 0.45} \\
\midrule
\raisebox{1.30ex}{\strut 5} & \shortstack{\strut \underline{33163}\\[-0.85ex]\strut 0.24} & \shortstack{\strut \textbf{32884}\\[-0.85ex]\strut 0.24} & \shortstack{\strut 265016\\[-0.85ex]\strut 2.81} & \shortstack{\strut 263581\\[-0.85ex]\strut 9.69} & \shortstack{\strut 744170\\[-0.85ex]\strut 7.03} & \shortstack{\strut 1151553\\[-0.85ex]\strut 6.06} & \shortstack{\strut 171107\\[-0.85ex]\strut 3.01} \\
\midrule
\raisebox{1.30ex}{\strut 6} & \shortstack{\strut \underline{597536}\\[-0.85ex]\strut \textbf{2.56}} & \shortstack{\strut \textbf{593600}\\[-0.85ex]\strut \underline{2.59}} & \shortstack{\strut 34237542\\[-0.85ex]\strut 629.84} & \raisebox{1.30ex}{\strut -} & \shortstack{\strut 21395668\\[-0.85ex]\strut 253.25} & \shortstack{\strut 22280726\\[-0.85ex]\strut 120.53} & \shortstack{\strut 2472620\\[-0.85ex]\strut 22.56} \\
\midrule
\raisebox{1.30ex}{\strut 7} & \shortstack{\strut \underline{9893500}\\[-0.85ex]\strut \textbf{22.24}} & \shortstack{\strut \textbf{9825774}\\[-0.85ex]\strut \underline{26.40}} & \raisebox{1.30ex}{\strut -} & \raisebox{1.30ex}{\strut -} & \raisebox{1.30ex}{\strut -} & \raisebox{1.30ex}{\strut -} & \shortstack{\strut 38677276\\[-0.85ex]\strut 306.90} \\
\bottomrule
\end{tabular}
\label{tab:u4-b2-fd}
\end{table}

\begin{table}[t]
\centering
\small
\caption{The circuit sizes and compilation times of \texttt{u4-b2-so}}
\begin{tabular}{c|ccccccc}
\toprule

n & ours & ours* & bella-cd & bella-ka & bella-ph & d4 & c2d \\
\midrule
\raisebox{1.30ex}{\strut 4} & \shortstack{\strut \underline{1372}\\[-0.85ex]\strut 0.02} & \shortstack{\strut \textbf{1322}\\[-0.85ex]\strut 0.02} & \shortstack{\strut 4543\\[-0.85ex]\strut 0.07} & \shortstack{\strut 4539\\[-0.85ex]\strut 0.08} & \shortstack{\strut 19758\\[-0.85ex]\strut 0.15} & \shortstack{\strut 17366\\[-0.85ex]\strut 0.05} & \shortstack{\strut 3954\\[-0.85ex]\strut 0.49} \\
\midrule
\raisebox{1.30ex}{\strut 5} & \shortstack{\strut \underline{8336}\\[-0.85ex]\strut 0.04} & \shortstack{\strut \textbf{7984}\\[-0.85ex]\strut 0.04} & \shortstack{\strut 21065\\[-0.85ex]\strut 0.28} & \shortstack{\strut 20959\\[-0.85ex]\strut 0.32} & \shortstack{\strut 60210\\[-0.85ex]\strut 0.71} & \shortstack{\strut 24122\\[-0.85ex]\strut 0.12} & \shortstack{\strut 18011\\[-0.85ex]\strut 1.02} \\
\midrule
\raisebox{1.30ex}{\strut 6} & \shortstack{\strut \underline{61664}\\[-0.85ex]\strut \textbf{0.12}} & \shortstack{\strut \textbf{58382}\\[-0.85ex]\strut \underline{0.15}} & \shortstack{\strut 149234\\[-0.85ex]\strut 1.52} & \shortstack{\strut 149334\\[-0.85ex]\strut 2.02} & \shortstack{\strut 389499\\[-0.85ex]\strut 4.30} & \shortstack{\strut 148295\\[-0.85ex]\strut 0.62} & \shortstack{\strut 106311\\[-0.85ex]\strut 2.07} \\
\midrule
\raisebox{1.30ex}{\strut 7} & \shortstack{\strut \underline{502300}\\[-0.85ex]\strut \textbf{0.61}} & \shortstack{\strut \textbf{471000}\\[-0.85ex]\strut \underline{0.96}} & \shortstack{\strut 1291983\\[-0.85ex]\strut 7.08} & \shortstack{\strut 1411551\\[-0.85ex]\strut 12.14} & \shortstack{\strut 1227252\\[-0.85ex]\strut 7.06} & \shortstack{\strut 1126991\\[-0.85ex]\strut 3.26} & \shortstack{\strut 818082\\[-0.85ex]\strut 5.77} \\
\bottomrule
\end{tabular}
\label{tab:u4-b2-so}
\end{table}

\begin{table}[t]
\centering
\small
\caption{The circuit sizes and compilation times of \texttt{u4-b2-ao}}
\begin{tabular}{c|ccccccc}
\toprule

n & ours & ours* & bella-cd & bella-ka & bella-ph & d4 & c2d \\
\midrule
\raisebox{1.30ex}{\strut 4} & \shortstack{\strut \underline{1484}\\[-0.85ex]\strut 0.03} & \shortstack{\strut \textbf{1412}\\[-0.85ex]\strut 0.03} & \shortstack{\strut 4034\\[-0.85ex]\strut 0.07} & \shortstack{\strut 3933\\[-0.85ex]\strut 0.08} & \shortstack{\strut 29488\\[-0.85ex]\strut 0.21} & \shortstack{\strut 32384\\[-0.85ex]\strut 0.07} & \shortstack{\strut 2930\\[-0.85ex]\strut 0.46} \\
\midrule
\raisebox{1.30ex}{\strut 5} & \shortstack{\strut \underline{11898}\\[-0.85ex]\strut 0.07} & \shortstack{\strut \textbf{11244}\\[-0.85ex]\strut 0.07} & \shortstack{\strut 30383\\[-0.85ex]\strut 0.36} & \shortstack{\strut 30621\\[-0.85ex]\strut 0.49} & \shortstack{\strut 502148\\[-0.85ex]\strut 4.62} & \shortstack{\strut 850598\\[-0.85ex]\strut 4.01} & \shortstack{\strut 17176\\[-0.85ex]\strut 0.99} \\
\midrule
\raisebox{1.30ex}{\strut 6} & \shortstack{\strut \underline{97008}\\[-0.85ex]\strut \textbf{0.26}} & \shortstack{\strut \textbf{91602}\\[-0.85ex]\strut \underline{0.28}} & \shortstack{\strut 256112\\[-0.85ex]\strut 2.05} & \shortstack{\strut 259476\\[-0.85ex]\strut 3.78} & \shortstack{\strut 252054\\[-0.85ex]\strut 1.65} & \shortstack{\strut 216020\\[-0.85ex]\strut 0.58} & \shortstack{\strut 118718\\[-0.85ex]\strut 1.88} \\
\midrule
\raisebox{1.30ex}{\strut 7} & \shortstack{\strut \underline{809096}\\[-0.85ex]\strut \textbf{1.49}} & \shortstack{\strut \textbf{765620}\\[-0.85ex]\strut \underline{1.83}} & \shortstack{\strut 3370331\\[-0.85ex]\strut 11.94} & \shortstack{\strut 2484503\\[-0.85ex]\strut 31.35} & \shortstack{\strut 3771598\\[-0.85ex]\strut 11.15} & \shortstack{\strut 2933345\\[-0.85ex]\strut 4.27} & \shortstack{\strut 981102\\[-0.85ex]\strut 5.74} \\
\bottomrule
\end{tabular}
\label{tab:u4-b2-ao}
\end{table}

\begin{table}[t]
\centering
\small
\caption{The circuit sizes and compilation times of \texttt{u4-b2-fo}}
\begin{tabular}{c|ccccccc}
\toprule

n & ours & ours* & bella-cd & bella-ka & bella-ph & d4 & c2d \\
\midrule
\raisebox{1.30ex}{\strut 4} & \shortstack{\strut \underline{2204}\\[-0.85ex]\strut 0.05} & \shortstack{\strut \textbf{2168}\\[-0.85ex]\strut 0.05} & \shortstack{\strut 4238\\[-0.85ex]\strut 0.06} & \shortstack{\strut 4238\\[-0.85ex]\strut 0.06} & \shortstack{\strut 4237\\[-0.85ex]\strut 0.05} & \shortstack{\strut 4579\\[-0.85ex]\strut 0.02} & \shortstack{\strut 3396\\[-0.85ex]\strut 0.59} \\
\midrule
\raisebox{1.30ex}{\strut 5} & \shortstack{\strut 15664\\[-0.85ex]\strut 0.09} & \shortstack{\strut \textbf{15432}\\[-0.85ex]\strut 0.09} & \shortstack{\strut 18287\\[-0.85ex]\strut 0.23} & \shortstack{\strut 18287\\[-0.85ex]\strut 0.23} & \shortstack{\strut 18485\\[-0.85ex]\strut 0.22} & \shortstack{\strut 22681\\[-0.85ex]\strut 0.11} & \shortstack{\strut \underline{15557}\\[-0.85ex]\strut 1.29} \\
\midrule
\raisebox{1.30ex}{\strut 6} & \shortstack{\strut 120436\\[-0.85ex]\strut \textbf{0.28}} & \shortstack{\strut 118544\\[-0.85ex]\strut \underline{0.31}} & \shortstack{\strut 82250\\[-0.85ex]\strut 1.30} & \shortstack{\strut 82250\\[-0.85ex]\strut 1.30} & \shortstack{\strut \underline{81333}\\[-0.85ex]\strut 1.29} & \shortstack{\strut 105180\\[-0.85ex]\strut 0.51} & \shortstack{\strut \textbf{45967}\\[-0.85ex]\strut 1.93} \\
\midrule
\raisebox{1.30ex}{\strut 7} & \shortstack{\strut 967728\\[-0.85ex]\strut \textbf{1.28}} & \shortstack{\strut 951588\\[-0.85ex]\strut \underline{1.59}} & \shortstack{\strut \underline{303279}\\[-0.85ex]\strut 4.63} & \shortstack{\strut \underline{303279}\\[-0.85ex]\strut 4.54} & \shortstack{\strut \underline{303279}\\[-0.85ex]\strut 4.38} & \shortstack{\strut 424793\\[-0.85ex]\strut 2.22} & \shortstack{\strut \textbf{166183}\\[-0.85ex]\strut 4.33} \\
\bottomrule
\end{tabular}
\label{tab:u4-b2-fo}
\end{table}

\subsection{Evaluation on OBDD}
\label{app:obdd_eval}

We further evaluate the OBDD representation induced by our compilation procedure.
Although our compiler is primarily designed for d-DNNF, the two-stage construction naturally imposes a type-based block order on the propositional variables: unary blocks are processed first, followed by binary blocks in lexicographic order.
As shown in~\cref{cor:OBDD}, for a fixed sentence, the resulting circuit can
be transformed into an OBDD respecting this order with only a constant-factor overhead. 
Thus we also evaluated the OBDD compilation over the same benchmarks.

We compare our approach with a generic grounding-based CNF-to-OBDD baseline implemented using \textsc{CUDD}, a widely used BDD package. 
For each sentence and domain size, as in the d-DNNF experiments, the \textsc{CUDD} baselines take as input the grounded propositional encoding of the corresponding \fotwo{} sentence.

Since OBDDs are canonical for a fixed boolean function and variable ordering, variable ordering is one of the dominant factors affecting both OBDD size and compilation time. 
We therefore evaluate two dynamic reordering strategies provided by \textsc{CUDD}.
The first baseline, denoted by \textsc{cudd-WIN3}, uses \texttt{CUDD\_REORDER\_WINDOW3}, a local window-based reordering heuristic, and is intended as a more runtime-oriented baseline.
The second baseline, denoted by \textsc{cudd-SIFT}, uses \texttt{CUDD\_REORDER\_SIFT}, the standard sifting heuristic, and is intended as a size-oriented baseline. 
All remaining experimental settings, including the hardware configuration, and the timeout of $1800$ seconds, are the same as in the d-DNNF evaluation. 
The results are summarized in Table~\ref{tab:obdd_res}.

Overall, the results show that the type block order used by our compiler is also effective for OBDD compilation. In particular, our method often scales to larger domain sizes and achieves substantially lower compilation times than the generic \textsc{CUDD} baselines. 
This advantage is consistent with the design of our compiler: instead of searching over propositional assignments directly, it branches on unary and binary types and uses efficient extendability checks to prune infeasible
contexts.

At the same time, the comparison also shows that our current ordering is not always optimal for OBDDs. 
The size-oriented \textsc{cudd-SIFT} baseline produces smaller OBDDs, albeit often with substantially higher compilation times.
This suggests that our block order captures useful first-order structure, but does not fully replace specialized BDD variable ordering optimization.
Designing OBDD specific order refinements that preserve the benefits of our method is therefore a natural direction for future work.

\begin{table}[t]
\centering
\small
\caption{
Summary of experimental results for OBDD compilation.
For each sentence, we report the maximum domain size $n$ solved within the time limit (1800\,s) by our compiler, \textsc{cudd-WIN3}, and \textsc{cudd-SIFT}, respectively.
For each \textsc{CUDD} baseline, we also report the OBDD-node-count ratio and compilation-time ratio relative to our compiler, measured at the largest domain size that both our compiler and the corresponding baseline solved.
}  
\begin{tabular}{c|c|cc}
\toprule
Sentence
& Max $n$
& \textsc{cudd-WIN3}
& \textsc{cudd-SIFT} \\

\midrule

$\sentence_{RB}$      & 20/18/15 & 0.571/40.987  & 0.876/473.459 \\
$\sentence_{E}$                & 23/20/16 & 0.790/3.756   & 0.691/71.469 \\
$\sentence_{RBE}$    & 13/13/11 & 0.999/9.221   & 0.431/159.125 \\
$\sentence_{P}$                   & 12/12/14 & 0.188/0.298   & 0.007/0.067 \\
$\sentence_{D}$                 & 15/8/11  & 8.075/599.273 & 0.606/130.086 \\

\midrule

\texttt{u2-b2}                  & 10/9/8   & 1.055/3.725   & 1.391/254.123 \\
\texttt{u4-b1}                  & 10/9/7   & 1.771/14.939  & 2.616/233.924 \\
\texttt{u4-b2}                  & 8/7/6    & 1.275/12.174  & 0.846/25.030 \\
\texttt{u4-b3}                  & 7/6/6    & 1.507/7.504   & 1.171/224.939 \\
\texttt{u6-b2}                  & 8/6/5    & 7.590/98.618  & 9.043/144.609 \\

\midrule

\texttt{u4-b2-ad}               & 8/7/6    & 1.486/23.870  & 0.450/9.726 \\
\texttt{u4-b2-ao}               & 8/6/6    & 1.774/43.535  & 0.365/48.075 \\
\texttt{u4-b2-cc}               & 8/7/6    & 1.424/17.239  & 0.833/25.690 \\
\texttt{u4-b2-fd}               & 8/7/6    & 1.563/12.291  & 0.979/18.328 \\
\texttt{u4-b2-so}               & 8/7/6    & 0.716/48.632  & 0.294/36.599 \\
\texttt{u4-b2-fo}               & 8/8/7    & 0.096/1.182   & 0.025/3.829 \\

\bottomrule
\end{tabular}
\label{tab:obdd_res}
\end{table}

\end{document}